\documentclass[sigconf, nonacm]{acmart}
\usepackage[utf8]{inputenc}
\usepackage{amsmath}
\usepackage{amsthm}
\usepackage{amsfonts}
\usepackage{natbib}
\usepackage{bm}
\usepackage{thm-restate}
\usepackage{tikz}
\usetikzlibrary{calc,shapes.misc}
\usepackage{pgfplots}
\pgfplotsset{compat=newest}

\usepackage{wrapfig}
\usepackage{subcaption}
\usepackage{float}

\allowdisplaybreaks

\def\reals{{\mathbb R}}

\newcommand{\eps}{\epsilon}

\newcommand{\band}{\text{band}}

\usepackage{framed}
\usepackage{setspace}
\usepackage{listings}
\usepackage[linesnumbered, ruled,vlined]{algorithm2e}

\def\pr{\mathbb{P}}
\newcommand{\ind}{\mathds{1}}

\def\calD{\mathcal{D}}
\def\calE{\mathcal{E}}

\def\calN{\mathcal{N}}

\def\calQ{\mathcal{Q}}
\def\calR{\mathcal{R}}
\def\calS{\mathcal{S}}

\def\calX{\mathcal{X}}

\def\calY{\mathcal{Y}}
\def\reals{{\mathbb R}}
\def\nats{{\mathbb N}}
\def\integers{{\mathbb Z}}
\def\norm#1{\mathopen\| #1 \mathclose\|}

\newcommand{\arxiv}[1]{#1}

\theoremstyle{plain}
\newtheorem{theo}{Theorem}[section]
 
\newtheorem{theorem}[theo]{Theorem}
\newtheorem{lemma}[theo]{Lemma}
\newtheorem{claim}[theo]{Claim}
\newtheorem{corollary}[theo]{Corollary}

\theoremstyle{definition}
\newtheorem{definition}[theo]{Definition}
\newtheorem{example}[theo]{Example}
\newtheorem{proposition}[theo]{Proposition}

\theoremstyle{remark}
\newtheorem{remark}[theo]{Remark}

\usepackage{dsfont}

\DeclareMathOperator*{\argmax}{argmax}

\DeclareMathOperator*{\Lap}{Lap}
\DeclareMathOperator*{\Quantile}{Quantile}
\DeclareMathOperator*{\Insert}{Insert}

\DeclareMathOperator*{\Compress}{Compress}
\DeclareMathOperator*{\rank}{rank}
\DeclareMathOperator*{\ix}{ix}
\DeclareMathOperator*{\val}{val}

\DeclareMathOperator*{\dpexpgk}{\texttt{DPExpGK}}
\DeclareMathOperator*{\dpexpgkGumb}{\texttt{DPExpGKGumb}}
\DeclareMathOperator*{\dphistgk}{\texttt{DPHistGK}}

\DeclareMathOperator*{\dpexpfull}{\texttt{DPExpFull}}

\def\bx{\text{\textbf{x}}}

\def\Gumb{\text{Gumb}}

\def\maxIndex{\text{maxIndex}}
\def\maxValue{\text{maxValue}}

\begin{document}

\title{Bounded Space Differentially Private Quantiles}


\author{Daniel Alabi}
\affiliation{Harvard University}
\email{alabid@g.harvard.edu}
\author{Omri Ben-Eliezer}
\affiliation{Massachusetts Institute of Technology}
\email{omrib@mit.edu}
\author{Anamay Chaturvedi}
\affiliation{Northeastern University}
\email{chaturvedi.a@northeastern.edu}

\begin{abstract}
Estimating the quantiles of a large dataset is a fundamental problem in both the streaming algorithms literature and the differential privacy literature. However, all existing private mechanisms for distribution-independent quantile computation require space at least linear in the input size $n$. In this work, we devise a differentially private algorithm for the quantile estimation
problem, with strongly sublinear space complexity,
in the one-shot and continual observation settings. 
Our basic mechanism estimates any $\alpha$-approximate quantile of a length-$n$ stream over a data universe $\mathcal{X}$ with probability $1-\beta$ using
$O\left( \frac{\log (|\calX|/\beta) \log (\alpha \epsilon n)}{\alpha \epsilon} \right)$
space while
satisfying $\eps$-differential privacy 
at a single time point. 
Our approach builds upon deterministic streaming algorithms
for non-private quantile estimation 
instantiating the exponential mechanism using a utility function defined on sketch items, while (privately)
sampling from intervals defined by the sketch. We also present
another algorithm based on histograms that is especially
suited to the multiple quantiles case.
We implement our algorithms
and experimentally evaluate them on synthetic and real-world datasets. 
\end{abstract}
\maketitle

\section{Introduction}
Quantile estimation is a fundamental subroutine in data analysis and statistics. For $q \in [0,1]$, the $q$-quantile in a dataset of size $n$ is the element ranked $\lceil qn\rceil$ when the elements are sorted from smallest to largest.
Computing a small number of quantiles in a possibly huge collection of data elements can serve as a quick and effective sketch of the ``shape'' of the data. 
Quantile estimation also serves an essential role in robust statistics, where data is generated from some distribution but is contaminated by a non-negligible fraction of outliers, i.e., ``out of distribution'' elements that may sometimes even be adversarial. 
For example, the median (50th percentile) 
of a dataset is used as a robust estimator of the mean in such situations where the data may be contaminated.
Location parameters can also be (robustly)
estimated via truncation or
winsorization, an operation that relies on quantile
estimation as a subroutine~\citep{Tukey60, Huber64}.
Rank-based nonparametric
statistics can be used for hypothesis testing (e.g., the Kruskal-Wallis 
test statistic~\citep{KW52}). Thus, designing quantile-based or rank-based
estimators, 
whether distribution-dependent or distribution-agnostic, 
is important in many scenarios.

Maintaining the privacy of individual users or data items, or even of groups, is an essential prerequisite in many modern data analysis and management systems.
Differential privacy (DP) is a rigorous and now well-accepted definition of privacy for data analysis and machine learning.
In particular, there is already a substantial amount of literature on differentially private quantile estimation
(e.g., see~\cite{NissimRS07, AsiD20,GJK21,TzamosVZ20}).\footnote{
Robust estimators are also known to be useful for accurate
differentially private estimation; see, e.g., the work of Dwork and Lei~\cite{DworkL09} in the context of quantile estimation for the
interquartile range and for medians.} 

All of these previous works, however, either make certain distributional assumptions about the input, or assume the ability to access all input elements (thus virtually requiring a linear or worse space complexity). Such assumptions may be infeasible in many practical scenarios, where large scale databases have to quickly process streams of millions or billions of data elements without clear a priori distributional characteristics.

The field of streaming algorithms aims to provide space-efficient algorithms for data analysis tasks. These algorithms typically maintain good accuracy and fast running time while having space requirements that are substantially smaller than the size of the data. While distribution-agnostic quantile estimation is among the most fundamental problems in the streaming literature \cite{AgarwalCHPWY13, FelberO17, GreenwaldK01, HungDing10, KarninLL16, MankuRL99, MunroPaterson80, ShrivastavaBAS04, WangLYC13}, no differentially private sublinear-space algorithms for the same task are currently known. 
Thus, the following question, essentially posed 
 by~\cite{Smith11} and~\cite{MirMNW11},  naturally arises:
\begin{quote}
    \textit{Can we design differentially private quantile estimators that use space 
    sublinear in the stream length, have efficient running time, provide high-enough utility, and do not rely on restrictive distributional assumptions?}
\end{quote}

It is well known \cite{MunroPaterson80} that exact computation of quantiles cannot be done with sublinear space, even where there are no privacy considerations. Thus, one must resort to approximation. Specifically, for a dataset of $n$ elements, an \emph{$\alpha$-approximate $q$-quantile} is any element which has rank $(q \pm \alpha)n$ when sorting the elements from smallest to largest, and it is known that the space complexity of $\alpha$-approximating quantiles is $\tilde\Omega(1/\alpha)$~\cite{MunroPaterson80}. In our case, the general goal is to efficiently compute $\alpha$-approximate quantiles in a (pure or approximate) differentially private manner.

\subsection{Our Contributions}

We answer the above question 
affirmatively by providing theoretically proven algorithms with accompanying
experimental validation for quantile estimation with DP guarantees. The algorithms are suitable for private computation of either a single quantile or multiple quantiles. Concretely, the main contributions are:
\begin{enumerate}
\item We devise $\dpexpgk$, a differentially private sublinear-space algorithm for quantile estimation based on the exponential mechanism. In order to achieve sublinear space complexity, our algorithm carefully instantiates the exponential mechanism with the basic blocks being intervals from the Greenwald-Khanna \cite{GreenwaldK01} data structure for non-private quantile
estimation, rather than single elements.  We prove general distribution-agnostic utility bounds on our algorithm and show that the space complexity is logarithmic in $n$.  
\item We present $\dphistgk$, another differentially private mechanism for quantile estimation, which applies the Laplace mechanism to a histogram, again using intervals of the GK-sketch as the basic building block. We theoretically demonstrate that $\dphistgk$ may be  useful in cases where one has prior knowledge on the input.
\item We extend our results to the continual release
setting.
\item We empirically validate our results by evaluating $\dpexpgk$, analyzing and comparing various aspects of
performance on real-world and synthetic datasets.
\end{enumerate}

$\dpexpgk$ can be interpreted as a more ``general purpose'' solution that performs well unconditionally of the data characteristics. It is especially suitable for the single quantile problem, and can be adapted to multiple quantiles by splitting the privacy budget and applying standard composition theorems in differential privacy. 
On the other hand, $\dphistgk$ inherently solves the all-quantiles problem and
may be suitable when there are not too many bins (small set of possible values),
the variance of the target distribution is small, or
the approximation parameter (i.e., $\alpha$) is large.

\section{Related Work}
\subsection{Quantile Approximation of Streams and Sketches}

Approximation of quantiles in large data streams (without privacy guarantees) is among the most well-investigated problems in the streaming literature \cite{WangLYC13, GreenwaldK16, XiangD0Z20}. A classical result of Munro and Paterson from 1980 \cite{MunroPaterson80} shows that computing the median exactly with $p$ passes over a stream requires $\Omega(n^{1/p})$ space, thus implying the need for approximation to obtain very efficient (that is, at most polylogarithmic) space complexity. 
Manku, Rajagopalan and Lindsay \cite{MankuRL99} built on ideas from \cite{MunroPaterson80} to obtain a randomized algorithm with only $O((1/\alpha) \log^{2}(n\alpha))$ for $\alpha$-approximating all quantiles; a deterministic variant of their approach with the same space complexity exists as well \cite{AgarwalCHPWY13}. 
The best known deterministic algorithm is that of Greenwald and Khanna (GK) \cite{GreenwaldK01} on which we build on in this paper, with a space complexity of $O(\alpha^{-1} \log(\alpha n))$ to sketch all quantiles for $n$ elements (up to rank approximation error of $\pm \alpha n$). A recent deterministic lower bound of Cormode and Vesel\'{y} \cite{CormodeVesely2020}  (improving on \cite{HungDing10}) shows that the GK algorithm is in fact optimal among deterministic (comparison-based) sketches. 

Randomization and sampling help for streaming quantiles, and the space complexity becomes independent of $n$; an optimal space complexity of $O((1/\alpha) \log \log(1/\beta))$ was achieved by Karnin, Lang and Liberty \cite{KarninLL16} (for failure probability $\beta$), concluding a series of work on randomized algorithms \cite{AgarwalCHPWY13, FelberO17, LuoWYC16, MankuRL99}.

The problem of biased or relative error quantiles, where one is interested in increased approximation accuracy for very small or very large quantiles, has also been investigated \cite{CKLTV20,CormodeKMS06}; it would be interesting to devise efficient differentially private algorithms for this problem. 

Recall that our approach is based on the Greenwald-Khanna deterministic all-quantiles sketch \cite{GreenwaldK01}. While some of the aforementioned randomized algorithms have a slightly better space complexity, differential privacy mechanisms are inherently randomized by themselves, and the analysis seems somewhat simpler and more intuitive when combined with a deterministic sketch. This, of course, does not rule out improved private algorithms based on modern efficient randomized sketches
\arxiv{(see Section \ref{sec:conclusion})}. 

\subsection{Differential Privacy}

\paragraph{Differentially private single quantile estimation:}{In the absence of distributional assumptions, in work by Nissim, Raskhodnikova and Smith \cite{NissimRS07} and Asi and Duchi \cite{AsiD20}, the trade-off between accuracy and privacy is improved by scaling the noise added for obfuscation in an instance-specific manner for median estimation. Another work by Dwork and Lei \cite{DworkL09} uses a ``propose-test-release" paradigm to take advantage of local sensitivity; however, as observed in \cite{GJK21}, in practice the error incurred by this method is relatively large as compared to other works like \cite{TzamosVZ20}. The work \cite{TzamosVZ20} achieves the optimal trade-off between privacy and utility in the distributional setting, but again as observed by \cite{GJK21}, with a time complexity of $O(n^4)$, this method does not scale well to large data sets.}

A very recent work of Gillenwater et al.~\cite{GJK21} 
shows how to optimize the division of the privacy budget to estimate $m$ quantiles in a time-efficient manner. For estimation of $m$ quantiles, their time and space complexity are $O(mn \log (n) + m^2 n)$ and $O(m^2 n)$, respectively. 
They do an extensive experimental analysis and find lower error compared to previous work. However, although they provide intuition for why their method should incur relatively low error, they do not achieve formal theoretical accuracy guarantees. Note that their results substantially differ from ours, as they hold for private \emph{exact} quantiles and so cannot achieve sublinear space. 
Kaplan et al.~\citep{KSS21} improve upon~\citep{GJK21}
for the multiple approximate quantiles problem by
using a tree-based (recursive) approach to CDF estimation.
All these works do not deal with quantile estimation in
the sublinear space settings.

\paragraph{Differentially private statistical estimation:}{ Estimation of global data statistics (or more generally, inference) is in general an important use-case of differential privacy~\cite{KV17}. The related private histogram release problem was studied by \cite{BalcerV18, cardoso2021differentially, aumuller2021differentially, LHRMM09} and in the continual observation privacy setting \cite{DworkNPR10, ChanSS11, ChenMHM17}. Perrier et al.~\cite{PerrierAK19} make improvements over prior work for private release of some statistics such as the moving average of a stream when the data distribution is light-tailed. They use private quantile estimation of the input stream as a sub-routine to achieve their improvements. B\"ohler and Kerschbaum \cite{BohlerK20} solve the problem of estimating the joint median of two private data sets with time complexity sub-linear in the size of the data-universe and provide privacy guarantees for small data sets as well as limited group privacy guarantees unconditionally against polynomially time-bounded adversaries.}

\paragraph{Inherent privacy:}{ Another line of work \cite{BlockiBDS12, Smith0T20, ChoiDKY20} demonstrates that sketching algorithms for streaming problems might have inherent privacy guarantees under minimal assumptions on the dataset in some cases. For such algorithms, relatively little noise needs to be added to preserve privacy unconditionally.}



\section{Preliminaries and Notation}
In this section we give standard differential privacy notation and formally describe the quantile estimation problem. We also present the Greenwald-Khanna sketch guarantees in a form that is adapted for our use.

\subsection{Differential Privacy}

\begin{definition}[Differential Privacy~\cite{DworkMNS06}]
Let $\calQ:\calX^n\rightarrow\calR$ be a (randomized) mechanism.
For any $\eps\geq 0, \delta\in[0, 1]$, $\calQ$ satisfies
\textbf{$(\eps, \delta)$-differential privacy} if for any
neighboring databases $\bx\sim\bx'\in\calX^n$ and any
$S\subseteq\calR$,
$$
\pr[\calQ(\bx)\in S] \leq e^\eps\pr[\calQ(\bx')\in S] + \delta.
$$
\end{definition}

The probability is taken over the coin tosses of $\calQ$.
We say that $\calQ$ satisfies pure differential privacy
($\eps$-DP) if $\delta=0$
and approximate differential privacy ($(\eps, \delta)$-DP) 
if $\delta > 0$. We can set $\eps$ to be a small constant
(e.g., between 0.01 and 2) but will require that
$\delta\leq n^{-\omega(1)}$ be cryptographically small.
Stability-based histograms can be used to obtain approximate DP guarantees.
However, most of our results apply only to the pure DP setting.

\subsection{Quantile Approximation}

\begin{figure}
\begin{tikzpicture}[scale=0.4,
     roundnode/.style={circle, draw=black!60, thick, minimum size=7mm},
     roundnodemedapp/.style={circle, draw=blue!60, fill=blue!5, very thick, dashed, minimum size=7mm},
     roundnodemedapptrue/.style={circle, draw=green!100, fill=green!20, very thick, dashed, minimum size=7mm}
     ]
      \draw[ultra thick, <->] (0,0) -- (20,0);
      \foreach \x in {0,1,2,...,8}
        {        
          \coordinate (A\x) at ($(2,0)+(2*\x,0)$) {};
          \draw ($(A\x)+(0,5pt)$) -- ($(A\x)-(0,5pt)$);
          \ifthenelse{\(\x>1\)\AND\(\x<6\)}{
            \node at ($(A\x)-(0,3ex)$) {\Large \x};
          }{
          }
        }
      \node (domain) at (10,-2) {\large Totally ordered domain $\mathcal{X}$};
      \node[roundnode] at (4,2) {1} ;
      \node[roundnodemedapp] at (6,2) {2} ;
      \node[roundnodemedapp] at (6,4) {2} ;
      \node[roundnodemedapptrue] at (6,6) {2};
      \node[roundnodemedapp] at (8,2) {3} ;
      \node[roundnodemedapp] at (12,2) {5} ;
      \node[roundnode] at (12,4) {5} ;
      \node[roundnode] at (14,2) {6} ;
      \node at (10,8) {\Large Data set $S' = \{1,2,2,3,5,2,6,5 \}$};
    \end{tikzpicture}
    \caption{In the data set $S'$ of $8$ elements, the true $0.5$ quantile is the value $2$, and the values $2, 3, 4$ and $5$ are all acceptable $0.25$-approximate answers. Note that although $1$ is adjacent to $2$ in the data universe $\calX$, it is not an acceptable output.}
    \label{fig:example}
\end{figure}
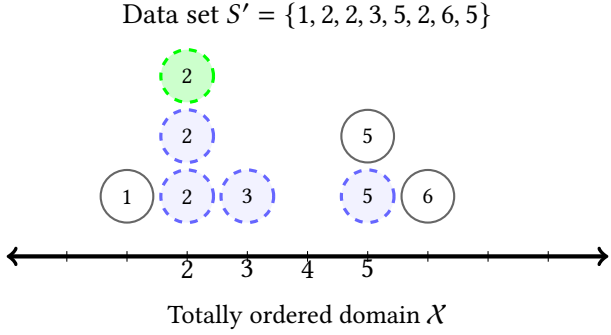

\begin{definition}\label{def:basics}
Let $X = ((x_1,1),\dots, (x_n,n))$ 
(sometimes implicitly referred to as $X = (x_1,\dots, x_n)$) be a stream of elements drawn from some finite totally ordered data universe $\calX$, i.e., $x_i \in \calX$ for all 
$i\in[n]$.
    \begin{enumerate}
    \item \textbf{Rank}: Given a totally ordered finite data universe $\calX$, a data set $X$ and a value $x \in \calX$, let $\rank_{X}(x) = \rank(X, x) = \sum_{y \in X} \ind[y \leq x]$.
    \item \textbf{Stream Prefix}: Given an input data stream $X = (x_1, \dots, x_n)$ we define the prefix up to the index $s$ element of this stream $X[1:s] := (x_1, \dots, x_s)$. Note that we can overload notation and treat $X[1:s]$ as a data set by ignoring the order in which the elements arrive.
    \item \textbf{Checkpoint}: A set $S\neq\emptyset$ is a set of checkpoints for stream $X$ if $\forall j\in S$,
	$\exists$ value $v_j$ that is a $(\alpha/2)$-approximate $q$-quantile
	for $X[1:j]$.
    \item For $(x_i,i) \in X$, let $\val((x_i,i)) = x_i$ and $\ix ((x_i,i)) = \sum_{j \leq i} \allowbreak 
    |\{(x_j,j): x_j < x_i \mbox{ or } x_j = x_i, j<i \}|$.
    \item For $v_1, v_2 \in X$, we say that $v_1 \leq v_2$ if $\ix (v_1) \leq \ix (v_2)$.
    \item The $q$-quantile of $X$ is $\val (v)$ for $v \in X$ such that $\ix(v) = \lceil qn \rceil$.
    \item For $x \in \calX$, we define $r_{\min} (x) = |\{v \in X: \val(v) < x\}|$, $r_{\max} (x) = |\{v \in X: \val(v) \leq x\}|$ and  $\rank(X,x)$ to be the interval $[r_{\min}(x) , r_{\max} (x)]$.
    \item We say that $x \in \calX$ is an $\alpha$-approximate $q$-quantile for $X$ if $\rank(X,x) \cap [\lceil qn \rceil - \alpha n, \lceil qn \rceil + \alpha n] \not= \emptyset$.
    \end{enumerate}
With this notation, the data set $X$ is naturally identified as a multi-set of elements drawn from $\calX$.
\end{definition}

\arxiv{
\begin{example}
Given the data set $\{1, 2, 2, 3, 5, 2, 6, 5\}$ (refer to figure~\ref{fig:example}), the $0.5$ quantile is $2$, which we distinguish from the median, which would be the average of the elements ranked $4$ and $5$, i.e, $2.5$ for this data set. A straightforward way to obtain quantiles is to sort the dataset and the pick the element at the $\lceil q\cdot n\rceil$ position. This method only works in the offline (non-streaming) setting. For $\alpha = 0.25$, the $\alpha$-approximate $0.5$ quantiles are $2,3,4$ and $5$. Note that $4$, which does not occur in the data set, is still a valid response, but $1$, which occurs in the data set and is even adjacent to the true $0.5$ quantile $2$ in the data universe $\calX$, is not a valid response.
\end{example}

\begin{definition}[Single Quantile]
Given sample $S = (x_1, \ldots, x_n)$ in a streaming fashion in
arbitrary order, construct a data structure for computing the quantile
$Q_\calD^q$ such that for any $q\in(0, 1)$, with probability at least
$1-\beta$,
$$|Q_\calD^q - \tilde{Q}_S^q|\leq \alpha.$$
\end{definition}

\begin{definition}[All Quantiles]
Given sample $S = (x_1, \ldots, x_n)$ in a streaming fashion in
arbitrary order, construct a data structure for computing the quantile
$Q_\calD^q$ such that
with probability at least
$1-\beta$, for all values of $q\in M$ where $M\subset(0, 1)$,
$$|Q_\calD^q - \tilde{Q}_S^q|\leq \alpha.$$
\end{definition}
}

\arxiv{
In Figures~\ref{fig:np1a} and~\ref{fig:np1b}, we show 
the performance of the
(non-private) Greenwald-Khanna sketch.
As can be seen in Figure~\ref{fig:np1a}, for higher $\alpha$
(corresponding to bigger approximation and smaller space), the sketch
is less accurate than for lower $\alpha$ (smaller approximation
and larger space) (Figure~\ref{fig:np1b}).
In the private algorithms, the loss in approximation must also be
incurred.

\begin{figure}
  \begin{subfigure}[b]{0.2\textwidth}
    \includegraphics[width=\textwidth]{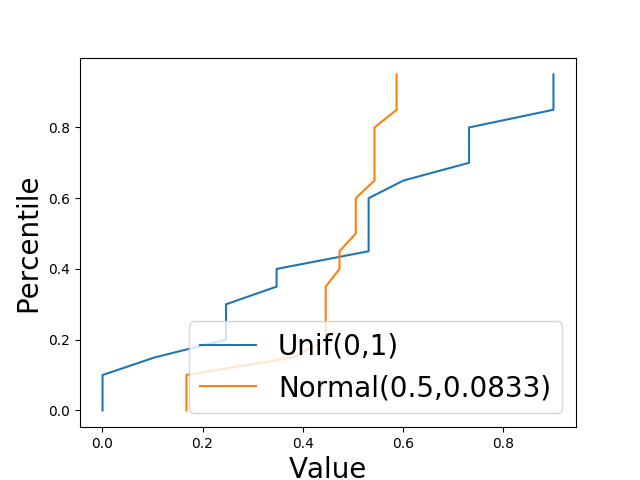}
    \caption{(Non-Private) $\alpha$-Approximate Quantiles for $\alpha = 0.1$ and $q\in(0, 1)$.}
    \label{fig:np1a}
  \end{subfigure}%
  \qquad
  \begin{subfigure}[b]{0.2\textwidth}
    \includegraphics[width=\textwidth]{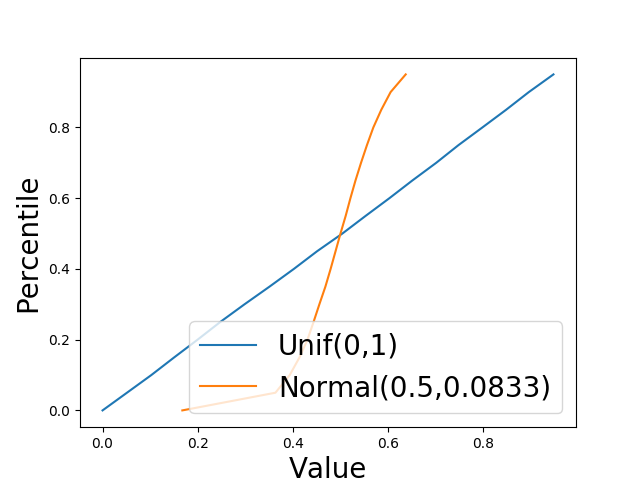}
    \caption{(Non-Private) $\alpha$-Approximate Quantiles for $\alpha = 0.0001$ and $q\in(0, 1)$.}
    \label{fig:np1b}
  \end{subfigure}
  \caption{}
\end{figure}
}

\subsection{Non-private Quantile Streaming}

\begin{lemma}[GK guarantees \cite{GreenwaldK01}]\label{lem:GKguarantee}
    Given a data stream $x_1, \dots, x_n$ of elements drawn from $\calX$, the GK sketch outputs a list $S(X)$ of $s (= O((1/\alpha) \log(\alpha n)))$ many tuples $(v_i, g_i, \Delta_i) \in (\calX \times \nats) \times \nats \times \nats$ for $i = 1, \dots, s$ such that if $X$ is the data multi-set then
    \begin{enumerate}
        \item $\rank (X, \val(v_i)) \subset [\sum_{j \leq i} g_j , \Delta_{i} + \sum_{j\leq i} g_j]$.
        \item $g_i + \Delta_i \leq 2\alpha n$.
        \item The first tuple is $(\min \{x \in X\}, 1, 0)$ and the last tuple is $(\max \{x \allowbreak \in X\}, 1, 0)$.
        \item The $v_i$ are sorted in ascending order. WLOG, the lower confidence interval bounds $\sum_{j \leq i} g_j$ and upper confidence interval bounds $\Delta_i + \sum_{j \leq i} g_j$ are also sorted in increasing order.
    \end{enumerate} 
\end{lemma}

\begin{proof}
    The first three statements are part of the GK sketch guarantee. For the third statement, i.e., to see that the $v_i$ are sorted in ascending order, we see that the GK sketch construction ensures that $\val(v_i) \leq \val(v_{i+1})$ for all $i$. Since an insertion operation always inserts a repeated value after all previous occurrences and the tuple order is always preserved, it follows that $\ix (v_i) \leq \ix(v_{i+1})$ as well, so in sum $\rank(X,v_i) \leq \rank(X,v_{i+1})$. In other words, the sort order in the GK sketch is \emph{stable}.
    
    The fact that the sequence $\sum_{j \leq i} g_j$ is sorted in increasing order follows from the non-negativity of the $g_i$. To ensure that $\Delta_i + \sum_{j \leq i} g_j$ are sorted in increasing order note that we always have that $\rank(X,v_i) \leq \rank(X,v_{i+1})$ so we can decrement $\Delta_i$ and ensure that $\Delta_{i} + \sum_{j \leq i} g_j \leq \Delta_{i+1} + \sum_{j \leq i+1} g_j $ without violating the guarantees of the GK sketch.
\end{proof}

\begin{remark}\label{rem:infty}
 We can add two additional tuples $(-\infty, 0,0)$ and $(\infty, 0,0)$ to the sketch, which corresponds to respective rank intervals $[0,0]$ and $[n+1,n+1]$. The bounds $g_i + \Delta_i \leq 2\alpha n$ are preserved. This will ensure that the sets $\{i : \val (v_i) < x \}$ and $\{i : \val (v_i) > x \}$ for any $x\in \calX$ are always non-empty.
\end{remark}

\section{Differentially Private Algorithms}
In this section we present our two DP mechanisms for quantile estimation.
Throughout, we assume that $\alpha$ is a user-defined approximation parameter. The goal is to obtain $\alpha$-approximate
$q$-quantiles.

The new algorithms we introduce are: 

\begin{enumerate}
\item $\dpexpgkGumb$ (Algorithm~\ref{alg:dpexpgkgumb}): An exponential mechanism based $(\eps, 0)$-DP 
algorithm for computing a single $q$-quantile. To solve the all-quantiles problem with approximation factor $\alpha$, one can run this algorithm iteratively with target quantile $0, \alpha, 2\alpha, \dots$. Doing so requires scaling the privacy parameter in each call by an additional $\alpha$ factor which increases the space complexity by a factor of $1/\alpha$. 
We extend our results to the continual observation setting \cite{DworkNPR10, ChanSS11}.

\item $\dphistgk$: A histogram based $(\eps, 0)$-DP 
algorithm (Algorithm~\ref{alg:dphistgk}) for the $\alpha$-approximate all-quantiles problem. The privacy guarantee of this algorithm is unconditional, but there is no universal theoretical utility bound as in the previous algorithm. However, in some cases the utility is provably better:
for example, we show that if the data set is drawn from a normal distribution (with unknown mean and variance), we can avoid the quadractic $1/\alpha^2$ factor in the sample complexity that we incur when using $\dpexpgkGumb$ for the same all-quantiles task. 
\end{enumerate}
The rest of the section contains the descriptions of the two algorithms, as well as a detailed theoretical analysis. The next section demonstrates how the histogram-based mechanism can be effectively applied to data generated from a normal distribution.

\subsection{\texorpdfstring{$\dpexpgkGumb$}{DPExpGKGumb}: Exponential Mechanism Based Approach}

We first establish how the exponential mechanism and the GK sketch may be used in conjunction to solve the single quantile problem. Concretely, the high level idea is to call the privacy preserving exponential mechanism with a utility function derived from the GK sketch. The exponential mechanism is a fundamental privacy primitive which when given a public set of choices and a private score for each choice outputs a choice that with high probability has a score close to optimal whilst preserving privacy. In the course of constructing our algorithms, we have to resolve two problems; one, how to usefully construct a utility function to pass to the exponential mechanism so that the private value derived is a good approximation to the $q$-quantile, and two, how to execute the exponential mechanism efficiently on the (possibly massive) data universe $\calX$. On resolving the first issue we get the (not necessarily efficient) routine Algorithm~\ref{alg:dpexpgk}, and on resolving the second issue we get an essentially equivalent but far more efficient routine i.e. Algorithm~\ref{alg:dpexpgkgumb}.

\paragraph{Constructing a score function:}{We recall that the GK sketch returns a short sequence of elements from the data set with a deterministic confidence interval for their ranks and the promise that for any target quantile $q \in [0,1]$, there is some sketch element that lies within $\alpha n$ units in rank of $\lceil qn \rceil$. One technicality that we run into when trying to construct a score function on the data universe $\calX$ is that when a single value occurs with very high frequency in the data set, the ranks of the set of occurrences can span a large interval in $[0,n]$, and there is no one rank we can ascribe to it so as to compare it with the target rank $\lceil qn \rceil$. This can be resolved by defining the score for any data domain value in terms of the distance of its respective rank interval $[r_{\min} (x), r_{\max} (x)]$ (formalized in Definition~\ref{def:basics}) from $\lceil qn \rceil$; elements whose intervals lie closer to the target have a higher score than those whose intervals lie further away. }

\paragraph{Efficiently executing the exponential mechanism:}{The exponential mechanism samples one of the public choices (in our case some element from the data universe $\calX$) with probability that increases with the quality of the choice according to the utility function. In general every element can have a possibly different score and the efficiency of the exponential mechanism can vary widely depending on the context. In our setting, the succinctness of the GK sketch leads to a crucial observation: by defining the score function via the sketch, the data domain is partitioned into a relatively small number of sets such that the utility function is constant on each partition. Concretely, for any two successive elements in the GK sketch, the range of values in the data universe that lie between them will have the same score according to our utility function. We can hence first sample a partition from which to output a value, and then choose a value from within that interval uniformly at random. To make our implementation even more efficient and easy to use, we also make use of the \emph{Gumbel-max} trick that allows us to iterate through the set of choices instead of storing them in memory.}

\paragraph{Outline:}{From Definition~\ref{def:rankEst} to Lemma~\ref{lem:rankEst}, we formalize how the GK sketch may be used to construct rank interval estimates for any data domain value. We then recall and apply the exponential mechanism with a utility function derived from the GK sketch (Definition~\ref{def:exp} to Lemma~\ref{lem:sensitivity} and Algorithm~\ref{alg:dpexpgk}), and derive the error guarantee Lemma~\ref{lem:acc_bound}. We conclude this subsection with a detailed description of an efficient implementation of the exponential mechanism (Algorithm~\ref{alg:dpexpgkgumb} and Lemma~\ref{lem:EMtoGumb}), and summarize our final accuracy and space complexity guarantees in Theorem~\ref{thm:utility_universal_bound}.}

\begin{definition}\label{def:rankEst}
    Let $\hat{r}_{\min} (x) = \max \{ \sum_{j \leq i} g_j : \val(v_i) < x \}$ and $\hat{r}_{\max} (x) = \min \{ \Delta_{i} + \sum_{j \leq i} g_j : \val(v_i) > x \}$. Note that for every $v \in X$ such that $\val (v) = x$, $\ix(v) \in [\hat{r}_{\min} (x),\allowbreak \hat{r}_{max} (x)]$.
\end{definition}

We formalize the rank interval estimation in a partition-wise manner as below.

\begin{lemma}\label{lem:simpleEst}
    Given a GK sketch $(v_1, g_1, \Delta_1), \dots, (v_s, g_s, \Delta_s)$, for every $x \in \calX$ one of the following two cases holds:
    \begin{enumerate}
        \item $x = v_i$ for some $i \in [s]$ and $i = \min \{j : v_j = x\}$,
        \begin{align*}
            \hat{r}_{\min} (x) &= \sum_{j \leq i} g_j\\ 
            \hat{r}_{\max} (x) &=  \min \{ \Delta_{i^*} + \sum_{j \leq i^*} g_j : \exists i^*, \val(v_{i^*}) > \val(v_{i}) \}
        \end{align*}
            
        \item $x \in (v_{i-1}, v_i)$, i.e., $x > v_{i-1}$ and $x< v_i$ for some $i \in [s]$,
        \begin{align*}
            \hat{r}_{\min} (x) &= \sum_{j \leq i} g_j \\
            \hat{r}_{\max} (x) &=  \Delta_{i+1} + \sum_{j \leq i+1} g_j
        \end{align*}
    \end{enumerate}
\end{lemma}

\begin{proof}
    Recall that by Remark~\ref{rem:infty} we always have that the first tuple and the last tuple are formal elements at $-\infty$ and $\infty$, ensuring that every data universe element either explicitly occurs in the GK sketch or lies between two values that occur in the GK sketch. Both statements now follow directly from Definition~\ref{def:rankEst} and the fact that the values $v_i$ occur in increasing order in the sketch (Lemma~\ref{lem:GKguarantee}).  
\end{proof}
 
 The quality of the rank interval estimate $[\hat{r}_{\min}(x),\hat{r}_{\max}(x)]$ compared to the true rank interval $[r_{\min}(x),r_{\max}(x)]$ is formalized as follows.

\begin{lemma}\label{lem:rankEst}
    $|r_{\min} (x) - \hat{r}_{\min} (x)| \leq 2\alpha n$ and $|r_{\max} (x) - \hat{r}_{\max} (x)| \leq 2\alpha n$.
\end{lemma}
\begin{proof}
    Let $i^* = \argmax_{i:\val(v_i) < x} \sum_{j\leq i} g_j$. Then by Definition of $i^*$, we have that $\val(v_{i^*}) < x \leq \val (v_{i^* +1})$. It follows that
    \begin{align*}
        [\ix(v_{i^*}), \ix(v_{i^*+1})] &\subset [\sum_{j\leq i^*} g_j, \Delta_{i^*+1} + g_{i^*+1} + \sum_{j\leq i^*} g_j] \\
        &\subset [\hat{r}_{\min} (x), \Delta_{i^*+1} + g_{i^*+1} + \hat{r}_{\min} (x)]
    \end{align*}
    Since $r_{\min} (x) \in [\ix(v_{i^*}), \ix(v_{i^*+1})]$ and $g_{i^* + 1 } + \Delta_{i^* + 1} \leq 2\alpha n$, it follows that $|r_{\min} (x) - \hat{r}_{\min} (x)| \leq 2\alpha n$. The other inequality follows analogously.
\end{proof}

\begin{definition}[Exponential Mechanism~\citep{McSherryT07}]
Let $u:\calS^S\times\calR\rightarrow\reals$ be an arbitrary utility function
with global sensitivity $\Delta_u$. For any database summary
$d\in\calS^S$ and privacy parameter $\eps > 0$, the exponential mechanism 
$\calE_u^\eps : \calS^S \to \calR$ outputs $r\in\calR$ with  probability 
$\propto\exp(\frac{\eps\cdot u(S(X), r)}{2\Delta_u})$ where
$$\Delta_u = \max_{X\sim X', r}|u(S(X), r) - u(S(X), r)|.$$

\label{def:exp}
\end{definition}

The following statement formalizes the trade-off between the privacy parameter $\epsilon$ and the tightness of the tail bound on the score attained by the exponential mechanism.

\begin{theorem}[~\citep{McSherryT07, Smith11}]
The exponential mechanism (Definition~\ref{def:exp}) satisfies $\eps$-differential
privacy. Further, the following tail bound on the utility holds:
\begin{align*}
    P\left(u(S(X), \calE_u^\eps (S(X))) < \max_{r \in R} u(S(X), r) - \frac{2 \Delta_u (t + \ln s)}{\eps} \right) \leq e^{-t},
\end{align*}
where $s$ is the size of the universe from which we are sampling from.
\label{thm:exp}
\end{theorem}

To run the exponential mechanism using our approximate rank interval estimates, we define a utility function as follows.

\begin{definition} \label{def:utility}
    Let $d(\cdot, \cdot)$ denote the $\ell_1$ metric on $\reals$. Given a sketch $S(X)$, we define a utility function on $\calX$:
    \begin{align*}
        u(S(X),x) &= - \min \{ |y - \lceil qn \rceil| : y \in [\hat{r}_{\min} (x), \hat{r}_{\max} (x)]  \} \\
        &= - d(\lceil qn \rceil, [\hat{r}_{\min} (x), \hat{r}_{\max} (x)])
    \end{align*}
\end{definition}

The magnitude of the noise that is added in the course of the exponential mechanism depends on the sensitivity of the score function, which we bound from above as follows.

\begin{lemma}\label{lem:sensitivity}
    For all $n > 1/\alpha$, the sensitivity of $u$ (i.e., $\Delta_u$)
    is at most $4 \alpha n + 2$ units.
\end{lemma}

\begin{proof}
    Fix any data set $X'$ neighbouring $X$ under swap DP and let $[r'_{\min}(\cdot),\allowbreak r'_{\max}(\cdot)]$ be the rank ranges with respect to $X'$ for values in $\calX$. Let $[\hat{r'}_{\min}(\cdot),\allowbreak \hat{r'}_{\max}(\cdot)]$ denote the confidence interval derived from the GK sketch $S(X')$ for values in $\calX$. 
    \begin{claim}
        $|r_{\min} (x) - r'_{\min} (x)| \leq 2$, $|r_{\max} (x) - r'_{\max} (x)| \leq 2$.
    \end{claim}
    \begin{proof}
        These bounds follow directly from the Definition of $r_{\min}$ and $r_{\max}$; under swap DP at most two elements of the stream are changed which implies that the count of the sets defining these terms changes by at most 1 unit each for a total shift of $2$ units (in fact, this can be bounded by $1$ unit).
    \end{proof}
    We now prove the sensitivity bound.
    \begin{align*}
        u(S(X),x) &= - d(\lceil qn \rceil, [\hat{r}_{\min} (x), \hat{r}_{\max} (x)]) \\
        &\leq - d(\lceil qn \rceil, [r_{\min} (x), r_{\max} (x)] ) + 2 \alpha n \\
        &\leq - d(\lceil qn \rceil, [r'_{\min} (x), r'_{\max} (x)] ) + 2 \alpha n + 2 \\
        &\leq - d(\lceil qn \rceil, [\hat{r'}_{\min} (x), \hat{r'}_{\max} (x)] ) + 4 \alpha n + 2 \\
        &\leq u(S(X'),x) + 4\alpha n + 2.
    \end{align*}
    Swapping the positions of $X$ and $X'$, we get the reverse bound to complete the sensitivity analysis.
\end{proof}

\begin{algorithm}
    \KwData{$X = (x_1, x_2,\ldots, x_n)$}
    \KwIn{$\eps, \alpha\text{ (approximation parameter)}, q\in[0, 1]\text{ (quantile parameters)}$, $\Delta_u$}

    $S(X) = \{ (v_i, g_i, \Delta_i) : i \in [s] \} \leftarrow GK(X,\alpha)$\\
    
    Define utility function
    \begin{equation}
        u(S(X),x) = - \min \{ |y - \lceil qn \rceil| : y \in [\hat{r}_{\min} (x), \hat{r}_{\max} (x)]  \}. 
    \label{eq:u}
    \end{equation}
    where 
    \begin{align*}
        \hat{r}_{\min} (x) &= \max \{ \sum_{j \leq i} g_j : \val(v_i) < x \} \\
        \hat{r}_{\max} (x) &= \min \{ \Delta_{i} + \sum_{j \leq i} g_j : \val(v_i) > x \}
    \end{align*} 
    
    Choose and output $e \in \calX$ with probability
    $$\propto \exp\left(\frac{\eps}{2\Delta_u}\cdot u(S(X), e)\right).$$
    
    \caption{$\dpexpgk$: Exponential Mechanism DP Quantiles : High Level Description}
    \label{alg:dpexpgk}
\end{algorithm}

We can now derive a high probability bound on the utility that is achieved by Algorithm~\ref{alg:dpexpgk}.

\begin{lemma}\label{lem:acc_bound}
    If $\hat{x}$ is the value returned $\dpexpgk$ then with probability $1-\beta$,
    \begin{align*}
        d(\lceil qn \rceil, [\hat{r}_{\min} (\hat{x}), \hat{r}_{\max} (\hat{x})]) \leq 2\alpha n + \frac{2(4\alpha n + 2) \log (|\calX|/\beta)}{\epsilon}.
    \end{align*}
\end{lemma}

\begin{proof}
    By construction, Algorithm~\ref{alg:dpexpgk} is simply a call to the exponential mechanism with utility function $u(S(X),\cdot)$,
    Since for any target $q$-quantile, $\lceil qn \rceil$ lies in $[0,n]$ it follows that there is some $i^* \in s$ such that $\lceil qn \rceil \in [\sum_{j\leq i^*} g_j , \Delta_{i^* + 1} + g_{i^* + 1} + \sum_{j \leq i^*} g_j]$. It follows that $d(\lceil qn \rceil, [r_{\min}(\val(v_{i^*}), r_{\max}(\val(v_{i^*}))]) \leq 2\alpha n$ and that hence $\max(u (S(X), x)) \geq -2\alpha n$. If $x^*$ is the output of the exponential mechanism, then applying the utility tail bound we get that with probability $1-\beta$,
    \begin{align*}
        u(S(X),x^*) \geq -2 \alpha n - \frac{2 (4\alpha n + 2) \log (|\calX|/\beta)}{\epsilon}.
    \end{align*}
    By definition of $u$, the desired bound follows.
\end{proof}

\begin{algorithm}
\KwData{$X = (x_1, x_2,\ldots, x_n)$}
\KwIn{$\eps, \alpha\text{ (approximation parameter)}, q\in[0, 1]\text{ (quantile parameters)}$}
Build summary sketch $S(X)$
and let $s = |S(X)|$.\\
Let    $(v_i, g_i, \Delta_i) = S(X)[i]$\ for all $i \in [s]$\\
$\maxIndex = -1$\\
$\maxValue = -\infty$\\
\tcc{Iterating over tuple values $v_i$}
Let $i = 1$ \label{alg:dpexpgkgumb;line:pivotS} \\
\While{$i <= s$}{
    $\hat{r}_{\min} = \sum_{j \leq i} g_j$\\
    $\hat{r}_{\max} =  \min \{ \Delta_{i^*} + \sum_{j \leq i^*} g_j : \val(v_{i^*}) > \val(v_{i}) \}$\\
    $u_i = - \min \{ |y - \lceil qn \rceil| : y \in [\hat{r}_{\min}, \hat{r}_{\max}]  \}.$\\
    $f = \frac{\eps}{2} u_i $\\
    $\tilde{f} = f + \Gumb(0, 1)$\\
    \If {$\tilde{f} > \maxValue$} {
       $\maxIndex = (i,\mbox{tuple})$\\
       $\maxValue = \tilde{f}$
    }
    $i \leftarrow \min \{ j : v_j > v_i \}$ \label{alg:dpexpgkgumb;line:pivotE}
}
\tcc{Iterating over intervals  between tuples $\calX(v_{i-1},v_i) \subset \calX$}
Let $i = 1$ \label{alg:dpexpgkgumb;line:intervalS} \\
\While{$i <= s$}{
    \If {$\calX(v_{i-1},v_i)$ is not empty}{
        $\hat{r}_{\min} = \sum_{j \leq i} g_j$ \\
        $\hat{r}_{\max} =  \Delta_{i+1} + \sum_{j \leq i+1} g_j$\\
        $u_{i-1,i} = - \min \{ |y - \lceil qn \rceil| : y \in [\hat{r}_{\min}, \hat{r}_{\max}]  \}.$\\
        $f = \log(|\calX(v_{i-1},v_i)|) + \frac{\eps}{2} u_{i-1,i} $\\
        $\tilde{f} = f + \Gumb(0, 1)$\\
        \If {$\tilde{f} > \maxValue$} {
           $\maxIndex = (i,\mbox{interval})$\\          
           $\maxValue = \tilde{f}$
        }
    }
    $i \leftarrow i+1$ \label{alg:dpexpgkgumb;line:intervalE}
}
\uIf{$\maxIndex = (i,\mbox{tuple})$ for some $i \in {1,\dots, s}$}{
    \Return{$v_i$}
}
\ElseIf{$\maxIndex = (i,\mbox{interval})$ for some $i \in {1,\dots, s}$}{
    Pick $v\in \calX(v_{i-1},v_i)$ uniformly at random\\ 
    \Return{v}
}

\caption{$\dpexpgkGumb$: Implementing the Exponential Mechanism on $S(X)$ using
the Gumbel Distribution}
\label{alg:dpexpgkgumb}
\end{algorithm}

As discussed before, in general a naive implementation of the exponential mechanism as in Algorithm~\ref{alg:dpexpgk} would in general not be efficient. To resolve this issue, in Algorithm~\ref{alg:dpexpgkgumb} we take advantage of the partition of the data domain by the score function and the \emph{Gumbel-max} trick to implement the exponential mechanism without any higher-order overhead and return an $\alpha$-approximate $q$ quantile.
This trick has become 
a standard way to implement the exponential mechanism
over intervals/tuples.

\begin{lemma}\label{lem:EMtoGumb}
Algorithm~\ref{alg:dpexpgkgumb} implements the exponential mechanism with utility function $u(S(X),\cdot)$ on the data universe $\calX$ with space complexity $O(|S(X)|)$ (where $S(X)$ is the GK sketch) and additional time complexity $O(|S(X)|\log |S(X)|)$.
\end{lemma}

\begin{proof}
As noted in previous work~\citep{ALT17},
if $Z_1, \ldots, Z_N$ are drawn i.i.d. from standard Gumbel distribution,
$$
\pr\left[f_i + Z_i = \max_{j\in[N]}\{f_j + Z_j\}\right] = \frac{\exp(f_i)}{\sum_{j\in[N]}\exp(f_j)}, \forall i\in[N].
$$
We recall that when running the exponential mechanism on $\calX$, we want to sample the element $x \in \calX$ with probability $\propto \exp(\epsilon u(S(X),x)$. To implement the exponential mechanism via the identification with Gumbel $\argmax$ distribution above, we will simply compute the scores $u(S(X),x)$ and let $f_i = \epsilon\cdot  u(S(X),x)$. 

For $x \in \calX$ such that $x = v_i$ for some $i\in [s]$, Algorithm~\ref{alg:dpexpgkgumb} directly computes the scores according to Definition~\ref{def:utility} and Lemma~\ref{lem:simpleEst}; this is formalized by lines~\ref{alg:dpexpgkgumb;line:pivotS} to \ref{alg:dpexpgkgumb;line:pivotE} in the pseudo code.

For $x \in \calX$ which lie strictly between the tuple values $\{v_i : i \in [s] \}$, we proceed as follows. Fixing $i$, from Lemma~\ref{lem:simpleEst} we have that that for  $\calX(v_{i-1},v_i) := \{x \in \calX : x > v_{i-1}, x < v_i \}$, the rank confidence interval estimate is the same, i.e. $[\sum_{j\leq i-1} g_j, \Delta_i + \sum_{j\leq i} g_j]$. It follows from Definition~\ref{def:utility} that for all such domain values the the utility function score $u(S(X),\cdot)$ is equal; this is denoted $u_{i-1,i}$ in the pseudo code. By summing the probabilities for sampling individual domain elements, it follows that the likelihood of the exponential mechanism outputting some value from the set $\calX(v_{i-1},v_i)$ is $\propto |\calX(v_{i-1},v_i)| \exp(\epsilon u_{i-1,i}/2) = \exp(\epsilon u_{i-1,i}/2 + \log (|\calX(v_{i-1},v_i)|))$. This is formalized by lines~\ref{alg:dpexpgkgumb;line:intervalS} to \ref{alg:dpexpgkgumb;line:intervalE} in the pseudo code. 

Finally, if some interval is selected, then by outputting elements chosen uniformly at random, we ensure that the likelihood of $x \in \calX(v_{i-1}, v_i)$ being output is $\propto \frac{1}{|\calX(v_{i-1},v_i)|} \cdot \exp(\epsilon u_{i-1,i}/2 + \log (|\calX(v_{i-1},v_i)|) = \exp (\epsilon u_{i-1, i})$. Note that we do not need to account for ties in the Gumbel scores as the event $f_i + Z_i = f_j + Z_j$ for any $j\not= i$ has measure $0$. \footnote{As is usual in the privacy literature, we assume that the sampling of
the $Gumb(0, 1)$ distribution can be done on finite-precision 
computers~\citep{BalcerV18}. While the problem of formally dealing with rounding has not been settled in the privacy literature~\cite{Mironov12},
for any practical purpose it easily suffices to store the output of the Gumbel distribution using a few computer words.}

To bound the space and time complexity; we note that by the guarantees of the GK sketch, the size of the sketch $S(X)$ is $O((1/\alpha) \allowbreak \log \alpha n)$; we compute Gumbel scores by iterating over tuples and intervals of which there are at most $O(S(X))$-many of each, each computation takes at most $O(\log |S(X)|)$ time, and only the max score and index seen at any point is tracked in the course of the algorithm.
\end{proof}

We can now state and prove our main theorem in this section, proving utility bounds for $\alpha$-approximating quantiles through $\dpexpgk$ with sublinear space. 

\begin{theorem}
\label{thm:utility_universal_bound}
    Algorithm~\ref{alg:dpexpgkgumb} is $\epsilon$-differentially private. Let $\hat{x}$ be the value returned by Algorithm~\ref{alg:dpexpgkgumb} when initialized with target quantile $q$. The following statements hold:
    \begin{enumerate}
        \item Algorithm~\ref{alg:dpexpgkgumb} can be run with space complexity $O(1/\alpha \log \alpha n)$, such that with probability $1-\beta$ 
        \begin{align}
            d(\lceil qn \rceil, [\hat{r}_{\min} (\hat{x}), \hat{r}_{\max} (\hat{x})]) \leq 2\alpha n + \frac{2(4\alpha n + 2) \log (|\calX|/\beta)}{\epsilon} \label{eqn:guar1}
        \end{align}
        \item For $n > \frac{\log |\calX|/\beta}{12 \alpha \epsilon }$, Algorithm~\ref{alg:dpexpgkgumb} can be run with space complexity $O\left( (\alpha \epsilon)^{-1} \log (|\calX|/\beta) \log (\alpha \epsilon n) \right) \label{eqn:guar2}$ such that with probability $1-\beta$, $\hat{x}$ is an $\alpha$ approximate $q$-quantile.
    \end{enumerate}
\end{theorem}

Our dependence in $\alpha$, which for practical purposes is usually the most important term, is optimal. Very recent subsequent work by Kaplan and Stemmer~\citep{kaplan2021note} shows how to improve the dependence in other parameters if approximate (rather than pure) differential privacy is allowed, or if the stream length is large enough.



\begin{proof}
    The privacy guarantee of Algorithm~\ref{alg:dpexpgkgumb} follows from the privacy guarantee of the exponential mechanism and Lemma~\ref{lem:EMtoGumb}. The accuracy bound in equation~\ref{eqn:guar1} is simply a restatement of Lemma \ref{lem:acc_bound}. To derive the second statement, we substitute $\frac{\alpha \min \{ \epsilon,1 \}}{24 \log (|\calX|/\beta)}$ for the approximation parameter $\alpha$ in equation~\ref{eqn:guar1} and get
    \begin{align*}
            &d(\lceil qn \rceil, [\hat{r}_{\min} (\hat{x}), \hat{r}_{\max} (\hat{x})]) \\
            &\leq 2 \cdot \frac{\alpha \min \{ \epsilon,1 \} n}{24 \log (|\calX|/\beta)} + \frac{2(4 (\frac{\alpha \min \{ \epsilon,1 \}}{24 \log (|\calX|/\beta)}) n + 2) \log (|\calX|/\beta)}{\epsilon} \\
            &\leq \frac{\alpha n}{12 \log (|\calX|/\beta)} + \frac{\alpha n}{3} + \frac{4 \log | \calX |/\beta }{\epsilon} \\
            &\leq \frac{\alpha n}{12} + \frac{\alpha n}{3} + \frac{\alpha n}{3} \\
            &\leq \alpha n.
    \end{align*}
    The space complexity bound now follows directly from the space complexity bound derived in Lemma~\ref{lem:EMtoGumb}, the space complexity bound $O(1/\alpha \log \alpha n)$ for the GK sketch, and by substituting \newline $\frac{\alpha \min \{ \epsilon,1 \}}{24 \log (|\calX|/\beta)}$ for $\alpha$.
\end{proof}

\subsection{\texorpdfstring{$\dphistgk$}{DPHistGK}: Histogram Based Approach}

For methods in this section, we assume that we have $K \geq 1$
\emph{disjoint} bins each of width $w$ (e.g., $w = \alpha/2$).
These bins are used to construct a histogram.

Essentially, Algorithm~\ref{alg:dphistgk} builds
an empirical histogram based on the GK sketch, adds noise
so that the bin values satisfy $(\eps, 0)$-DP, and
converts this empirical histogram to an approximate
empirical CDF, from which the quantiles can be 
approximately calculated.

\begin{algorithm}
\KwData{$X = (x_1, x_2,\ldots, x_n)$}
\KwIn{$\eps, \alpha\text{ (approximation parameter)}, q\in[0, 1]\text{ (quantile parameters)}, w$}

Build summary sketch $S(X)$ where\\
$S(X) = \{ (v_i, g_i, \Delta_i) : i \in [s] \} \leftarrow GK(X,\alpha)$\\

\tcc{cell labels $a_i$ and counts $c_i = 0$}
Initialize data-agnostic (empty) histogram $Hist = \langle(a_i, c_i), \ldots\rangle$ with cell widths $w$

\For {$(v_i, g_i, \Delta_i) \in S(X)$} {
    Insert $g_i$ counts of $v_i$ into histogram $Hist$
}
$c = 0$\\
$H = []$\\
\For {$(a_i, c_i)\in Hist$} {
   $\tilde{c}_i = \max(0, c_i + \Lap(0, 2/\eps))$
   
   Append $(a_i, c + \tilde{c}_i)$ to $H$

   $c = c + \tilde{c}_i$
}

$r = \lceil q\cdot n\rceil$

\For {$(b, rank)\in H$} {
    \If {$r < rank$} {
        \Return $b$
    }
}

\tcc{return last element of $H$}

\Return $H[|H|-1]$

\caption{$\dphistgk$: Computing DP Quantiles in Bounded Space}
\label{alg:dphistgk}
\end{algorithm}

\begin{lemma}
Algorithm~\ref{alg:dphistgk} satisfies $(\eps, 0)$-DP.
\end{lemma}

\begin{proof}
For any $i\in[n]$, any item $x_i$  can belong in at most one bin. Plus, the global
sensitivity of the function that computes the empirical histogram is $2$, since changing
a single item can change the contents of at most two bins.

As a result, adding noise of $\Lap(0, 2/\eps)$ to each bin satisfies $\eps$-DP
by Theorem~\ref{thm:laplace}.
\end{proof}

\begin{theorem}[Laplace Mechanism~\cite{DworkMNS06}]
Fix $\eps > 0$ and any function $f:\calY^n\rightarrow\reals^K$. The Laplace mechanism
outputs
$$
f(y) + (L_1, \ldots, L_K),
$$
$L_1, \ldots, L_K\sim\Lap(0, GS_f/\eps)$ where $GS_f$ is the global sensitivity of the
function $f$. Furthermore, the mechanism satisfies
$(\eps, 0)$-DP.
\label{thm:laplace}
\end{theorem}

\section{Learning differentially private quantiles of a normal distribution with unknown mean}
We demonstrate one use case of $\dphistgk$ where the space complexity required 
improves upon the worst-case bound for $\dpexpgk$, Theorem \ref{thm:utility_universal_bound}. While the histogram based mechanism does not have universal utility bounds in the spirit of the above theorem, the results in this section serve as one simple example where it may yield desirable accuracy while using less space.

Suppose that we are given an i.i.d. sample $S = (X_1, X_2, \ldots, X_n)$ such that for all $i\in[n]$,
$X_i\sim\calN(\mu, \sigma^2I_{d\times d})$, $\mu\in\reals^d$. 
The goal is to estimate DP quantiles of the distribution $\calN(\mu, \sigma^2I_{d\times d})$ without
knowledge of $\mu$ or $\sigma^2$.
We will show how to estimate the quantiles assuming
that $\sigma^2$ is known. Note that it is easy to generalize the work
to the case where $\sigma^2$ is unknown as follows:

For any sample $S$ drawn from i.i.d. from 
$\calN(\mu, \sigma^2I_{d\times d})$, the $1-\beta$ confidence interval is
$$
\bar{X} \pm \frac{\sigma}{\sqrt{n}}\cdot z_{1-\beta/2},
$$
where $z_{1-\beta/2}$ is the $1-\beta/2$ quantile of the standard
normal distribution and $\bar{X}$ is the empirical mean. The length of this interval is fixed and equal to
$$
\frac{2\sigma z_{1-\beta/2}}{\sqrt{n}} = \Theta\left(\frac{\sigma}{\sqrt{n}}\sqrt{\log\frac{1}{\beta}}\right).
$$

In the case where $\sigma^2$ is unknown, the confidence interval becomes
$$
\bar{X} \pm \frac{s}{\sqrt{n}}\cdot t_{n-1, 1-\beta/2},
$$
where $s^2 = \frac{1}{n-1}\sum_{i=1}^n(X_i - \bar{X})^2$ is the 
sample variance (sample estimate of $\sigma^2$) and
$t_{n-1, 1-\beta/2}$ is the $1-\beta/2$ quantile of the $t$-distribution
with $n-1$ degrees of freedom.
The length of the interval can be shown to be
$$
\frac{2\sigma}{\sqrt{n}}\cdot k_n\cdot t_{n-1, 1-\beta/2} =
\Theta\left(\frac{\sigma}{\sqrt{n}}\sqrt{\log\frac{1}{\beta}}\right),
$$
where $k_n = 1 - O(1/n)$ is an appropriately chosen constant.
See~\citep{LehmannE, KV17, keener2010theoretical} for more details and discussion. We will assume that $\sigma^2$
is known and proceed to show sample and space complexity bounds.
One could also estimate the variance in a DP way and then prove the
complexity bounds.

For any $q\in(0, 1)$, we denote the $q$-quantile of the sample as $Q_S^q$ and the $q$-quantile of the distribution as
$Q_\calD^q$.

That is, for any $q\in(0, 1)$ and sample $S = (X_1, X_2, \ldots, X_n)$,
we wish to obtain a DP $q$-quantile $\tilde{Q}_S^q$ with the following guarantee:
$$
\pr[\norm{Q_\calD^q - \tilde{Q}_S^q}\geq \alpha] \leq \beta,
$$
for any $\beta\in(0, 1], \alpha > 0$.

We shall proceed to use a three-step approach:
(1) Estimate a DP range of the population in sub-linear space;
(2) Use this range of the population to construct a DP histogram using
the stream $S$;
(3) Use the DP histogram to estimate one or more quantiles via the
sub-linear data structure of Greenwald and Khanna.

In the case where $d=1$,
by Theorem~\ref{thm:normal}, there exists an
$(\eps, \delta)$-DP algorithm $\tilde{Q}_S^q$ such that
if $\mu\in(-R, R)$ then using space of
$O(\max\{\frac{R}{\sigma}, \frac{1}{\alpha}\log \alpha n\})$ (with
probability 1) as long as
the stream length is at least
\begin{align*}
n \geq O\bigg(
\max\bigg\{
\min\bigg\{&
O\left(\frac{R}{\eps\sigma\alpha}\log\frac{R}{\sigma\beta}\right),
O\left(\frac{R}{\eps\sigma\alpha}\log\frac{1}{\beta\delta}\right)
\bigg\},\\
&O\left(\frac{R^2}{\sigma^2\alpha^2}\log\frac{1}{\beta}\right)\bigg\}
\bigg),
\end{align*}
we get the guarantee that, 
for all $\beta\in(0, 1], \alpha > 0$,
$$
\pr[\norm{Q_\calD^q - \tilde{Q}_S^q}\geq \alpha] \leq \beta.
$$
Intuitively, this means that: 
(1) \textbf{Space}: We need less space to estimate any quantile with DP guarantees if the distribution is less concentrated
(i.e., $\sigma$ can be large) or if we do not require a high degree of accuracy
for our queries (i.e., $\alpha$ can be large).
(2) \textbf{Stream Length}: We need a large stream length to
estimate quantiles if we require a high degree of accuracy
(i.e., smaller $\beta, \alpha$), or do not have a good public estimate
of $\mu$ (large $R$), or have small privacy parameters 
(small $\eps, \delta$), or have concentrated datasets
(small $\sigma$).

\begin{theorem}
For the 1-D normal distribution $\calN(\mu, \sigma^2)$, let
$S = (X_1, X_2, \ldots, X_n)$ be a data stream
through which we wish to obtain $\tilde{Q}_S^q$, a DP estimate of
the $q$-quantile of the distribution.

For any $q\in(0, 1)$, there exists an $(\eps, \delta)$-DP algorithm
such that, with probability at least $1-\beta$, we
obtain $|Q_\calD^q - \tilde{Q}_S^q|\leq \alpha$ for any $\alpha > 0$,
$\beta\in(0, 1], \eps, \delta\in(0, 1/n)$ and for stream length
\begin{align*}
n \geq 
\max\left\{
\min\left\{
A,B
\right\},
C
\right\}, \ \ \text{where}
\end{align*}
$$
A = O\left(\frac{R}{\eps\sigma\alpha}\log\frac{R}{\sigma\beta}\right), B = O\left(\frac{R}{\eps\sigma\alpha}\log\frac{1}{\beta\delta}\right), C = O\left(\frac{R^2}{\sigma^2\alpha^2}\log\frac{1}{\beta}\right)
$$
as long as $\mu\in(-R, R)$ and using space of 
$O(\max\{\frac{R}{\sigma}, \frac{1}{\alpha}\log \alpha n\})$.
\label{thm:normal}
\end{theorem}

\begin{proof}
For any stream $S = (X_1, \ldots, X_n)$,
we use the triangle inequality so that
\begin{align}
|Q_\calD^q - \tilde{Q}_S^q| &\leq |Q_\calD^q - Q_S^q| + |Q_S^q - \tilde{Q}_S^q|\\
&\leq \alpha/2 + \alpha/2.
\end{align}

$|Q_\calD^q - Q_S^q|\leq \alpha/2$ follows with probability $1-\beta/2$ by Corollary~\ref{cor:converge} and
$|Q_S^q - \tilde{Q}_S^q|\leq \alpha/2$ follows with probability $1-\beta/2$ by Lemma~\ref{lem:dp}. The space complexity follows
with probability 1 via the deterministic nature of
the Greenwald-Khanna sketch.
\end{proof}

\begin{lemma}[Dvoretzky-Kiefer-Wolfowitz inequality~\cite{dvoretzky1956}]
For any $n\in\mathbb{Z}_+$, let $X_1, \ldots, X_n$ be i.i.d. random variables with
cumulative distribution function $F$ so that $F(x)$ is the probability that a single random variable
$X$ is less than $x$ for any $x\in\reals$. 
Let the corresponding empirical distribution function be 
$F_n(x) = \frac{1}{n}\sum_{i=1}^n\ind[X_i \leq x]$ for any $x\in\reals$.
Then for any $\gamma > 0$,
$$
\pr\left(\sup_{x\in\reals}|F_n(x) - F(x)| > \gamma\right)\leq 2\exp(-2n\gamma^2).
$$

\label{lem:dkw}
\end{lemma}

\begin{corollary}
For any $q\in(0, 1)$, let 
$Q_\calD^q$ be the $q$-quantile estimate for the distribution
$\calD$ and $Q_S^q$ be the $q$-quantile estimate for the sample.
Then,
$|Q_\calD^q - Q_S^q| \leq \alpha/2$ with probability $1-\beta/2$
when $n \geq \frac{2}{\alpha^2}\log4/\beta$.
\label{cor:converge}
\end{corollary}

\begin{proof}
Follows by the DKW inequality (Lemma~\ref{lem:dkw}) where $n\geq \frac{1}{2\gamma^2}\log4/\beta$ and
$\gamma = \alpha/2$.
\end{proof}

\begin{lemma}
For any $q\in(0, 1)$, $\alpha > 0$, $\beta\in(0, 1]$,
$\eps, \delta\in(0, 1/n)$, 
there exists an $(\eps, \delta)$-differentially
private algorithm $\tilde{Q}_S^q$
for computing the $q$-quantile such that
$$|Q_S^q - \tilde{Q}_S^q| \leq \alpha/2,$$
with probability $\geq 1 - \beta$ for stream length
$$
n \geq O\left(
\min\{
O\left(\frac{R}{\eps\sigma\alpha}\log\frac{R}{\sigma\beta}\right),
O\left(\frac{R}{\eps\sigma\alpha}\log\frac{1}{\beta\delta}\right)
\}
\right).
$$

Furthermore, with probability 1, $\tilde{Q}_S^q$ uses space of
$O(\max\{\frac{R}{\sigma}, \frac{1}{\alpha}\log \alpha n\})$.
\label{lem:dp}
\end{lemma}

\begin{proof}
First, by the tail bounds of the Gaussian distribution 
(Claim~\ref{claim:gauss}), we can obtain that
for any $i\in[n]$,
$$
\pr[|X_i - \mu| > c] \leq 2e^{-c^2/2\sigma^2},
$$
so that by the union bound,
$$
\pr[\exists i, |X_i - \mu|\geq c] \leq 2ne^{-c^2/2\sigma^2},
$$
which implies that for any $\beta\in(0, 1]$,
$$
\pr[\forall i, |X_i - \mu|\leq \sigma\sqrt{2\log 4n/\beta}] \geq 1 - \beta/2,
$$
which holds by our sample complexity (stream length) guarantees.

Next, let $r = \lceil R/\sigma \rceil$. 
\footnote{Note that this argument is similar to the arguments for
Algorithm 1 in~\cite{KV17}.}
Divide
$[-R - \sigma/2, R + \sigma/2]$ into $2r + 1$ bins of length
at most $\sigma$ each. Each bin $B_j$ should equal
$((j-0.5)\sigma, (j+0.5)\sigma]$ for any $j\in\{-r, \ldots, r\}$.
Next run the histogram learner of Lemma~\ref{lem:dphist} with
per-bin accuracy parameter of $\alpha/K$,
high-probability parameter of $\beta/2$,
privacy parameters $\eps, \delta\in(0, 1/n)$, and number of bins
$K = 2\lceil R/\sigma \rceil + 1$.
We can do this because of our sample complexity (stream length)
bounds.
Then we obtain
noisy estimates $\tilde{p}_{-r}, \ldots, \tilde{p}_{r}$ with per-bin
accuracy of $\alpha/K$. Then any quantile estimate would have
accuracy of $\alpha$ (by summing noisy estimates for at most $K$ bins).

Next, we use these bins to construct a sketch (private by DP
post-processing) based on the deterministic algorithms
of~\citep{GreenwaldK04} to, with probability 1, obtain space of
$O(\max\{\frac{R}{\sigma}, \frac{1}{\alpha}\log \alpha n\})$.
\end{proof}

\begin{lemma}[Histogram Learner~\cite{BunNSV15, Vadhan17, KV17}]
  For every $K\in\mathbb{N}\cup\{\infty\}$ and every collection of disjoint bins
  $B_1, \ldots, B_K$ defined on the domain $\calX$.
  For any $n\in\mathbb{N}$,
  $\eps, \delta\in (0, 1/n), \alpha > 0,$ and $\beta\in(0, 1)$, there exists an
  $(\eps, \delta)$-DP algorithm $M:\calX^n\rightarrow\reals^K$ such that for every distribution
  $\mathbb{D}$ on the domain $\calX$, if
  \begin{enumerate}
  \item $X_1, \ldots, X_n \sim \mathbb{D}$, $p_k = \pr[X_i\in B_k]$
  for any $k\in[K]$,
  \item $(\tilde{p}_1, \ldots, \tilde{p}_K) \leftarrow M(X_1, \ldots, X_n)$,
  \item $n \geq \max\left\{\min\left\{\frac{8}{\eps\alpha}\log\frac{2K}{\beta}, \frac{8}{\eps\alpha}\log\frac{4}{\beta\delta}\right\}, \frac{1}{2\alpha^2}\log\frac{4}{\beta}\right\}$,
  \end{enumerate}
  then (over the randomness of the data $X_1, \ldots, X_n$ and of $M$)
  \begin{enumerate}
  \item $\pr_{\underline{X}\sim\mathbb{D}, M}[|\tilde{p}_k - p_k|\leq\alpha] \geq 1-\beta$,
  \item $\pr[\argmax_k \tilde{p}_k = j] \leq np_j$ if $K\geq 2/\delta$,
  \item $\pr[\argmax_k \tilde{p}_k = j] \leq np_j + 2\exp(-(\eps n/8)\cdot (\max_k p_k))$ if $K < 2/\delta$.
  \end{enumerate}
\label{lem:dphist}
\end{lemma}

\begin{claim}[Gaussian Tail Bound]
Let $Z$ be a random variable distributed according to a
standard normal distribution (with mean 0 and variance 1).
For every $t > 0$,
$$
\pr[|Z| > t] \leq 2\exp(-t^2/2).
$$

\label{claim:gauss}
\end{claim}

\section{Continual Observation}
\newcommand{\cp}{\mathrm{cp}}
\newcommand{\GK}{\mathrm{GK}}

We now describe how our one-shot approach can be used as a black box to obtain a continual observation solution \cite{DworkNPR10, ChanSS11}.
    
    \begin{algorithm}
        \KwData{Input stream $X = (x_1, \dots, x_n)$ for some $n > n_{\min} =  \Omega \left(\frac{1}{\eps\alpha^2}\log^2 n\log\left(\frac{|\calX|\log n}{\alpha\beta}\right)\right)$, privacy parameter $\eps$, approximation parameter $\alpha$,
        target quantile $q$}
        Let $s = 0$\\
        $\epsilon^* \leftarrow O(\alpha \epsilon/\log n)$\\
        $\alpha^* \leftarrow \alpha/2$\\
        $\cp(s) \leftarrow n_{\min}$ \tcc{stores stream checkpoints}
        Instantiate $\GK \leftarrow GK(\alpha^*)$ \\
        $v_s \leftarrow \perp$ \tcc{holds $\alpha$-approximate $q$-quantile}
        \For{stream element $x_s\in X$}{
            $s \leftarrow s+1$\\
            $\GK.insert(x_s)$\\
            \uIf{$s = \lceil \cp(s-1) \rceil$}{
                $v_s \leftarrow \dpexpgk(\GK, \epsilon^*,
                \alpha^*, q)$ \\
                $\cp(s) \leftarrow \cp({s-1}) (1 + \alpha/2)$\\
            }
            \uElse{
                $v_s \leftarrow v_{s-1}$\\
                $\cp(s) \leftarrow \cp({s-1})$\\
            }
        }
        \caption{Continual Observation DP Quantiles}
        \label{alg:CO}
    \end{algorithm}

	\begin{lemma} \label{lem:geometric}
	If $x \in X$ is an $\alpha/2$-approximate $q$-quantile for a data set $X$ (for some $q \in [0,1]$), then it is an $\alpha$-approximate $q$-quantile for any data set $X' \supset X$ such that $|X'| \leq (1+\alpha/2)|X|$.
	\end{lemma}

	\begin{proof}
		Since $x$ is an $\alpha/2$-approximate $q$-quantile, we have that
		\begin{align*}
			(1 - \alpha/2) |X| \leq \rank_{X} (x) \leq  (1 + \alpha/2) |X|.
		\end{align*}
		We then have that
		\begin{align*}
			\rank_{X'} (x) &= \sum_{y \in X'} \ind[y \leq x] \\
			&= \sum_{y \in X} \ind[y \leq x] + \sum_{y \in X' \backslash X} \ind[y \leq x]\\
			\Rightarrow \sum_{y \in X} \ind[y \leq x] &\leq  \sum_{y \in X'} \ind[y \leq x] \leq \sum_{y \in X} \ind[y \leq x] + \sum_{y \in X' \backslash X} \ind[y \leq x] \\
			\Rightarrow (1 - \alpha/2) |X| &\leq  \sum_{y \in X'} \ind[y \leq x] \leq (1 + \alpha/2)|X| + (|X'| - |X|) \\
			\Rightarrow (1 - \alpha/2) |X| &\leq  \sum_{y \in X'} \ind[y \leq x] \leq (1 + \alpha/2)|X| + \alpha |X|/2 \\
			\Rightarrow (1 - \alpha) |X| &\leq  \sum_{y \in X'} \ind[y \leq x] \leq (1 + \alpha)|X|,
		\end{align*}
		i.e., $x$ is also an $\alpha$-approximate $q$-quantile for $X'$, as required.
	\end{proof}
	
	\begin{lemma}
	    For any $\beta\in(0, 1]$,
	    with probability $\geq 1-\beta$, 
	    for all $s \in \{ \cp(s') : s' \in [n] \}$, 
	    $v_{s}$ is an $\alpha/2$-approx.~$q$-quantile for
	    $X[1 : s]$.
	\end{lemma}
	
	\begin{proof}
	    First we bound the size of the set of checkpoints $\{ \cp(s') : s' \in [n] \}$. Since a new checkpoint value is generated only when $s \geq (1 + \alpha/2) \cp (s-1)$, it follows that for any new checkpoint where
	    $\cp(s')\neq\cp(s'-1)$, we have
	    $\cp(s') \geq (1 + \alpha/2) \cp(s'-1)$. 
	    Using Theorem~\ref{thm:utility_universal_bound},
	    we set the first checkpoint value to
	    $\frac{\log |\calX|/\beta}{12 \alpha^* \epsilon^*}$.
	    \footnote{The last checkpoint might occur at 
	    $(1+\alpha/2)n$. We may ignore checkpoints past $n$.}
	    As a result, there are at most
	    $k \leq \log_{1+\alpha/2} n = O(\frac{1}{\alpha}\log n)$
	    checkpoints using the fact that for all $x > -1$,
	    $\frac{x}{1+x} \leq \log(1+x) \leq x$.

	    Since $v_{\cp(s)}$ is the output of $\dpexpgk$ given a GK sketch with accuracy parameter $\alpha^*$ and privacy parameter $\epsilon^*$ it follows that with probability $1 - \beta^*$ where $\beta^* = O(\alpha \beta/\log n)$, $\rank_{X[1:s']} \in (1-\alpha/2, 1+\alpha/2) qn$. The stated result follows by applying the union bound over all private approximate quantile computations at checkpoints.
	\end{proof}
	
	\begin{theorem}
     	Let $\eps, \alpha > 0$, $n\in\integers$.
     	For any $\beta\in(0, 1]$,
	    with probability $\geq 1-\beta$, 
Algorithm~\ref{alg:CO} maintains an $\alpha$-approximate $q$-quantile at every point $s$ in the data stream 
$s = \Omega \left(\frac{1}{\eps\alpha^2}\log^2 n\log\left(\frac{|\calX|\log n}{\alpha\beta}\right)\right)$. Furthermore,
Algorithm~\ref{alg:CO} satisfies $\epsilon$-DP.
	\end{theorem}
	\begin{proof}
	    To see that Algorithm~\ref{alg:CO} is $\epsilon$-DP, we observe that the output of this algorithm throughout the data stream can be summarized by its outputs at the checkpoints $\{\cp(s) : s \in [n] \}$ (the points in the stream at which a new checkpoint is reached and a new value released are known publicly, so this suffices for privacy analysis).  It follows that there is a choice of $\epsilon^* = \Omega(\alpha \epsilon/\log n)$ that gives us an $\epsilon$-DP mechanism.
	
	    We now prove the accuracy guarantee. For any arbitrary $s \geq \frac{\log |\calX|/\beta}{12 \alpha^* \epsilon^*}$, i.e., the default checkpoint value, the output $v_s$ of the algorithm will equal the value at the most recent checkpoint, i.e., $v_{\cp(s)}$.  Then, since $|X[1:s]| = s \leq (1+\alpha/2) \cp(s) \leq  (1 + \alpha/2) |X[1:\cp(s)]|$, by Lemma~\ref{lem:geometric} it follows that $\rank_{X[1:s]} (v_s) \in (1 - \alpha, 1 + \alpha) qn$, i.e., $v_s$ is an $\alpha$-approximate $q$-quantile for $X[1:s]$.
	\end{proof}
	
	We conclude by mentioning that our continual observation solution incurs roughly a $O(\log n / \alpha)$ overhead in the space complexity, which is in line with classical works in differential privacy and adversarially robust streaming \cite{BJWY20, DworkNPR10}.
	Concurrently, Stemmer and Kaplan \cite{kaplan2021note} developed a notion of \emph{streaming sanitizers} which yields a continual observation guarantee ``for free'', without incurring such an overhead over the one-shot case.

\section{Experimental Evaluation}
\label{sec:experiments}
\pgfplotstableread{
x         y    y-min  y-max
100000	0.000037	0.000030	0.000022
200000	0.000039	0.000036	0.000047
300000	0.000034	0.000031	0.000035
400000	0.000030	0.000027	0.000038
}{\GkexpSmallUnif}

\pgfplotstableread{
x         y    y-min  y-max
100000	0.000009	0.000008	0.000016
200000	0.000005	0.000005	0.000012
300000	0.000004	0.000004	0.000007
400000	0.000002	0.000002	0.000002
}{\FullSmallUnif}

\pgfplotstableread{
x         y    y-min  y-max
100000 	 0.000135 	 0.000129 	 0.000148
200000 	 0.000201 	 0.000187 	 0.000233
300000 	 0.000150 	 0.000141 	 0.000137
400000 	 0.000254 	 0.000148 	 0.000135
}{\GkexpSmallNorm}

\pgfplotstableread{
x         y    y-min  y-max
100000 	 0.000056 	 0.000052 	 0.000082
200000 	 0.000029 	 0.000027 	 0.000039
300000 	 0.000017 	 0.000016 	 0.000024
400000 	 0.000019 	 0.000016 	 0.000032
}{\FullSmallNorm}

\pgfplotstableread{
x         y    
100000 	 28575
200000 	 26565
300000 	 25296
400000 	 25011
}{\GkexpSmallUnifSizes}

\pgfplotstableread{
x         y    
100000 	 100000
200000 	 200000
300000 	 300000
400000 	 400000
}{\FullSmallUnifSizes}

\pgfplotstableread{
x         y    
100000 	 28656
200000 	 26646
300000 	 25398
400000 	 25053
}{\GkexpSmallNormSizes}

\pgfplotstableread{
x         y    
100000 	 100000
200000 	 200000
300000 	 300000
400000 	 400000
}{\FullSmallNormSizes}

\pgfplotstableread{
x         y    y-min  y-max
100000	0.040248	0.021147	0.022083
200000	0.041315	0.024133	0.022454
300000	0.040818	0.023345	0.024196
400000	0.038616	0.022057	0.023917
}{\GkexpLargeUnif}

\pgfplotstableread{
x         y    y-min  y-max
100000	0.000009	0.000008	0.000014
200000	0.000005	0.000005	0.000011
300000	0.000003	0.000002	0.000006
400000	0.000002	0.000002	0.000002
}{\FullLargeUnif}

\pgfplotstableread{
x         y    y-min  y-max
100000 	 0.211188 	 0.194443 	 0.192134
200000 	 0.210906 	 0.182216 	 0.198140
300000 	 0.215453 	 0.183121 	 0.188972
400000 	 0.212650 	 0.188389 	 0.187417
}{\GkexpLargeNorm}

\pgfplotstableread{
x         y    y-min  y-max
100000 	 0.000065 	 0.000057 	 0.000089
200000 	 0.000023 	 0.000021 	 0.000048
300000 	 0.000018 	 0.000015 	 0.000016
400000 	 0.000017 	 0.000015 	 0.000028
}{\FullLargeNorm}

\pgfplotstableread{
x         y    
100000 	 33
200000 	 33
300000 	 33
400000 	 33
}{\GkexpLargeUnifSizes}

\pgfplotstableread{
x         y    
100000 	 100000
200000 	 200000
300000 	 300000
400000 	 400000
}{\FullLargeUnifSizes}

\pgfplotstableread{
x         y    
100000 	 30
200000 	 30
300000 	 30
400000 	 30
}{\GkexpLargeNormSizes}

\pgfplotstableread{
x         y    
100000 	 100000
200000 	 200000
300000 	 300000
400000 	 400000
}{\FullLargeNormSizes}

\pgfplotstableread{
x   y   y-min   y-max
0.1 	 0.000199 	 0.000192 	 0.000328
0.5 	 0.000105 	 0.000095 	 0.000123
1 	 0.000118 	 0.000105 	 0.000073
5 	 0.000137 	 0.000034 	 0.000032
}{\GkexpSmallUnifEps}

\pgfplotstableread{
x   y   y-min   y-max
0.1 	 0.000207 	 0.000195 	 0.000373
0.5 	 0.000048 	 0.000044 	 0.000095
1 	 0.000024 	 0.000022 	 0.000025
5 	 0.000020 	 0.000019 	 0.000025
}{\FullSmallUnifEps}

\pgfplotstableread{
x   y   y-min   y-max
0.1 	 0.000567 	 0.000537 	 0.001105
0.5 	 0.000165 	 0.000136 	 0.000185
1 	 0.000148 	 0.000122 	 0.000133
5 	 0.000138 	 0.000126 	 0.000129
}{\GkexpSmallNormEps}

\pgfplotstableread{
x   y   y-min   y-max
0.1 	 0.000605 	 0.000581 	 0.001227
0.5 	 0.000087 	 0.000077 	 0.000143
1 	 0.000061 	 0.000056 	 0.000081
5 	 0.000031 	 0.000028 	 0.000029
}{\FullSmallNormEps}

\pgfplotstableread{
x   y   y-min   y-max
0.1 	 0.094414 	 0.073439 	 0.080084
0.5 	 0.107291 	 0.065611 	 0.068421
1 	 0.117251 	 0.066682 	 0.063510
5 	 0.110183 	 0.059223 	 0.068683
}{\GkexpLargeUnifEps}

\pgfplotstableread{
x   y   y-min   y-max
0.1 	 0.000217 	 0.000199 	 0.000357
0.5 	 0.000042 	 0.000037 	 0.000064
1 	 0.000026 	 0.000024 	 0.000040
5 	 0.000020 	 0.000019 	 0.000025
}{\FullLargeUnifEps}

\pgfplotstableread{
x   y   y-min   y-max
0.1 	 0.209547 	 0.185342 	 0.185598
0.5 	 0.205141 	 0.197839 	 0.188772
1 	 0.194986 	 0.162147 	 0.190140
5 	 0.194043 	 0.168914 	 0.176988
}{\GkexpLargeNormEps}

\pgfplotstableread{
x   y   y-min   y-max
0.1 	 0.000522 	 0.000507 	 0.001426
0.5 	 0.000091 	 0.000086 	 0.000190
1 	 0.000052 	 0.000045 	 0.000076
5 	 0.000024 	 0.000023 	 0.000026
}{\FullLargeNormEps}

\pgfplotstableread{
x   y   y-min   y-max
0.1	0.010178	0.002032	0.017535
0.5	0.008884	0.002736	0.015853
1	0.009588	0.001967	0.017143
5	0.009606	0.001988	0.016456
}{\TaxiTwo}

\pgfplotstableread{
x   y   y-min   y-max
0.1	0.003025	0.002519	0.003483
0.5	0.003004	0.002439	0.003508
1	0.003030	0.002495	0.003572
5	0.002998	0.002479	0.003486
}{\TaxiThree}

\pgfplotstableread{
x   y   y-min   y-max
0.1	0.000237	0.000049	0.000476
0.5	0.000118	0.000037	0.000163
1	0.000098	0.000041	0.000158
5	0.000091	0.000040	0.000148 
}{\TaxiFour}

\pgfplotstableread{
x   y   y-min   y-max
0.1	0.000047	0.000003	0.000103
0.5	0.000041	0.000006	0.000077
1	0.000044	0.000007	0.000079
5	0.000073	0.000063	0.000082
}{\TaxiFive}

\pgfplotstableread{
x   y   y-min   y-max
0.1	0.000030	0.000004	0.000067
0.5	0.000009	0.000002	0.000021
1	0.000007	0.000002	0.000014
5	0.000003	0.000001	0.000006
}{\TaxiFull}

\pgfplotstableread{
x         y   
Full 	 1710672
1E-2 	 1929
1E-3 	 19182
1E-4     142365
1E-5     940836
}{\TaxiSizes}

\pgfplotstableread{
x   y   y-min   y-max
0.1	0.029838	0.007035	0.061898
0.5	0.029229	0.005445	0.054302
1	0.031480	0.007561	0.058147
5	0.032375	0.007250	0.060064
}{\EthCOTwo}

\pgfplotstableread{
x   y   y-min   y-max
0.1	0.016639	0.012351	0.021676
0.5	0.015923	0.011719	0.021117
1	0.015704	0.011720	0.021162
5	0.014832	0.012137	0.017394	
}{\EthCOThree}

\pgfplotstableread{
x   y   y-min   y-max
0.1	0.000576	0.000106	0.001011
0.5	0.000531	0.000097	0.000989
1	0.000515	0.000117	0.001004
5	0.000512	0.000075	0.001003
}{\EthCOFour}

\pgfplotstableread{
x   y   y-min   y-max
0.1	0.000247	0.000089	0.000500
0.5	0.000238	0.000043	0.000467
1	0.000233	0.000047	0.000475
5	0.000220	0.000030	0.000488	
}{\EthCOFive}

\pgfplotstableread{
x   y   y-min   y-max
0.1	0.000139	0.000027	0.000249
0.5	0.000134	0.000028	0.000236
1	0.000071	0.000000	0.000215
5	0.000000	0.000000	0.000000
}{\EthCOFull}

In this section, we experimentally evaluate our sublinear-space exponential-based mechanism, $\dpexpgkGumb$. We study how well our algorithm
performs in terms of accuracy and space usage. 
Note that we proved space complexity and accuracy bounds for the algorithm; part of this section validates that the space complexity of the algorithm is indeed very small in practice, and that the accuracy is typically closely tied to the approximation parameter $\alpha$.


Our main baselines will be $\dpexpfull$ (without use of the GK sketch but with the use of DP) and the true quantile value. The true quantile values are used to compute \textit{relative error}, the absolute value of the difference between the estimated quantile and the true quantile, divided by the standard deviation of the data set.

\paragraph{\textbf{$\dpexpfull$}}
The use of the exponential mechanism to $\epsilon$-privately compute $q$-quantiles
without the use of the sketching data structure. The space used by
this algorithm is of order $\Theta(n)$.
\arxiv{See Section~\ref{sec:fullspace} for further implementation details.}
We experimentally validate $\dpexpgkGumb$ 
varying
the following parameters:
$\eps > 0$ (the privacy parameter),
$n$ (the stream length), and
$\alpha$ (the approximation parameter).
We show results on both synthetically-generated datasets and
time-series real-world datasets.

The real-world time-series datasets~\citep{DG17} are:
(1) \textbf{Taxi Service Trajectory}: A dataset from the UCI
machine learning repository describing
trajectories performed by all 442 taxis (at the time) in the city
of Porto in Portugal~\citep{Moreira-MatiasGFMD13}.
This dataset contains real-valued
attributes with about 1.5 million instances.
(2) \textbf{Gas Sensor Dataset}: A UCI repository dataset containing
recordings of 16 chemical sensors exposed to varying
concentrations of two gas mixtures~\citep{FONOLLOSA2015618}.
The sensor measurements are acquired
continuously during a 12-hour time range and contains about
4 millions instances.

\subsection{Synthetic Datasets}

In this section, we compare our methods on synthetically generated
datasets. We vary the following parameters:
\begin{enumerate}
\item\textbf{Privacy Parameter $\eps$}:
For each choice of $\eps \geq 0$, we run the DP algorithms 
over 100 trials and output means and confidence intervals over these trials.
\item\textbf{Stream Length $n$}: The size of the stream
$x_1, \ldots, x_n$ received by the sketching algorithm.
When we vary $n$, we fix $\eps = 1$.
\item\textbf{Approximation Parameter $\alpha$}: The
approximation factor used by the internal GK sketch in our solution. 
\item\textbf{Data Distribution}: We generate data from
uniform and Gaussian distributions.
We use a uniform distribution in range $[0, 1]$ 
(i.e., $U(0, 1)$) or a normal distribution
with mean $0$ and variance $1$ (i.e., $N(0, 1)$), clipped to the interval $[-10,10]$.\footnote{Clipping is required for the exponential mechanism, as it must operate on some bounded interval of values. In any case, we never expect to see samples from $N(0,1)$ that lie outside $[-10,10]$ for any practical purpose; the probability for any given sample to satisfy this is minuscule, at about $\approx 10^{-21}$.}
\end{enumerate}

We show results on space usage and relative error (both plotted mostly on logarithmic scales) 
from non-DP estimates, as we vary the parameters listed above. Here, the \emph{relative error} 
is the absolute difference between the DP estimate and the true non-DP quantile value divided by the standard deviation of the data set.
For plots that compare the relative error to the stream length
(e.g., Figures~\ref{fig:SmallUnifError} and~\ref{fig:SmallNormError}),
the relative error are on the y-axis while the
stream length or the privacy parameter are on the x-axis.
For plots that compare the data structure size to the stream length
(e.g., Figures~\ref{fig:SmallUnifSize} and~\ref{fig:SmallNormSize}), the sizes are on the y-axis
while the stream length is be on the x-axis.
We graph the mean relative error over 100 trials of the exponential mechanism per experiment, as well as the
$80\%$ confidence interval computed by taking the $10$th and $90$th percentile.
While the quantile to be estimated can also be varied,
we fix it to $q = 0.5$ (the median) throughout as the relative error and the space usage generally do not seem to
vary much with choice of $q$.


In Figures~\ref{fig:SmallUnif}
and~\ref{fig:SmallNorm}
we vary the stream length
$n$ for an approximation factor of $\alpha = 10^{-4}$. The streams are either
normally or uniformly distributed.
In Figures~\ref{fig:SmallUnifError} and~\ref{fig:SmallUnifSize}, we compare 
$\dpexpgkGumb$ (space strongly sublinear in $n$)
vs. $\dpexpfull$ (uses space of $O(n)$) in terms of space usage
and accuracy.
We verify experimentally
that $\dpexpfull$ will use more space than $\dpexpgkGumb$.
On the other hand, the space savings incurred by $\dpexpgkGumb$ are expected to
create some drop in accuracy; recall that such a drop must occur even for non-private streaming algorithms.
Figures~\ref{fig:SmallNormError} and~\ref{fig:SmallNormSize} the streams
are from a normal distribution instead of uniform.

In general we find that although our method $\dpexpgkGumb$ incurs higher error, 
in absolute terms it remains quite small and the $95\%$ confidence intervals tend to be adjacent for $\dpexpgkGumb$ and $\dpexpfull$.
However, there is a clear trend of an exponential gap developing between their respective space usages which is a natural consequence
of the space complexity guarantee of the GK sketch. This holds for both distributions studied (Figures~\ref{fig:SmallUnif} and Figure~\ref{fig:SmallNorm}).

In Figures~\ref{fig:LargeUnif} and \ref{fig:LargeNorm}, we vary
the stream length for a relatively large approximation factor of 0.1. 
Here we see that compared to the
non-approximate method we incur far higher error, although there is also a 
concomitant increase in the space savings. This is not a typical 
use-case since the non-private error can itself be large, but we get a complete picture
of how space usage and performance vary with this user-defined parameter.

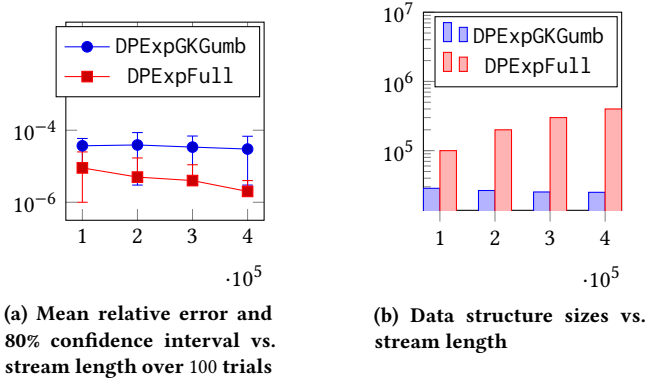
\begin{figure}
    \begin{subfigure}{.2\textwidth}
        \begin{tikzpicture}
            \begin{axis}[
            baseline=(current bounding box.center),
            height=120,width=120,,
            ymode = log,
            ymax = 0.1,
            bar width=10000,
            xtick={100000,200000,300000,400000,500000},
            ytick={1e-6,1e-4},
            ]
            
            \addplot+[error bars/.cd,
            y dir=both,y explicit]
            table [y error plus=y-max, y error minus=y-min] {\GkexpSmallUnif};
            \addplot+[error bars/.cd,
            y dir=both,y explicit]
            table [y error plus=y-max, y error minus=y-min] {\FullSmallUnif};
            \legend{$\dpexpgkGumb$,$\dpexpfull$}
            \end{axis}
        \end{tikzpicture}
        \caption{Mean relative error and 80\% confidence interval vs. stream length over $100$ trials}
        \label{fig:SmallUnifError}
    \end{subfigure}
    \hfill
    \begin{subfigure}{.2\textwidth}
        \begin{tikzpicture}
            \begin{axis}[
            baseline=(current bounding box.center),
            height=120,width=120,,
            ymode = log,
            ybar=0pt,
            restrict y to domain=0:500000,
            ymax = 10000000,
            bar width=30000,
            xtick={100000,200000,300000,400000,500000},
            xtick pos=left,
            ytick pos=left,
            ]
            
            \addplot+[error bars/.cd,
            y dir=both,y explicit]
            table {\GkexpSmallUnifSizes};
            \addplot+[error bars/.cd,
            y dir=both,y explicit]
            table {\FullSmallUnifSizes};
            \legend{$\dpexpgkGumb$,$\dpexpfull$}
            \end{axis}
        \end{tikzpicture}
        \caption{Data structure sizes vs. stream length \newline \newline}
        \label{fig:SmallUnifSize}
    \end{subfigure}
    \caption{$\dpexpgkGumb$ versus $\dpexpfull$ for uniformly random data, $\alpha = 10^{-4}$, $\epsilon = 1$, $q = 0.5$}
    \label{fig:SmallUnif}
\end{figure}

\begin{figure}
    \begin{subfigure}{.2\textwidth}
    \centering
    \begin{tikzpicture}
        \begin{axis}[
        height=120,width=120,,
        ymode = log,
        ymax = 0.1,
        bar width=30000,
        xtick={100000,200000,300000,400000,500000},
        ytick={1e-6,1e-4},
        ]
        \addplot+[error bars/.cd,
            y dir=both,y explicit]
            table [y error plus=y-max, y error minus=y-min] {\GkexpSmallNorm};
            \addplot+[error bars/.cd,
            y dir=both,y explicit]
            table [y error plus=y-max, y error minus=y-min] {\FullSmallNorm};
    \legend{$\dpexpgkGumb$,$\dpexpfull$}
    \end{axis}
\end{tikzpicture}
    \caption{Mean relative error and 80\% confidence interval vs. stream length over $100$ trials}
    \label{fig:SmallNormError}
\end{subfigure}
    \hfill
    \begin{subfigure}{.2\textwidth}
        \begin{tikzpicture}
            \begin{axis}[
            baseline=(current bounding box.center),
            height=120,width=120,,
            ymode = log,
            ybar=0pt,
            ymax = 10000000,
            bar width=30000,
            xtick={100000,200000,300000,400000,500000},
            xtick pos=left,
            ytick pos=left,
            ]
            
            \addplot+[error bars/.cd,
            y dir=both,y explicit]
            table {\GkexpSmallNormSizes};
            \addplot+[error bars/.cd,
            y dir=both,y explicit]
            table {\FullSmallNormSizes};
            \legend{$\dpexpgkGumb$,$\dpexpfull$}
            \end{axis}
        \end{tikzpicture}
        \caption{Data structure sizes vs. stream length \newline}
        \label{fig:SmallNormSize}
    \end{subfigure}
    \caption{$\dpexpgkGumb$ versus $\dpexpfull$ for normally distributed data, $\alpha = 10^{-4}$, $\epsilon = 1$, $q = 0.5$}
    \label{fig:SmallNorm}
\end{figure}
        
\begin{figure}
\begin{subfigure}{0.2\textwidth}
    \centering
    \begin{tikzpicture}
            \begin{axis}[
            baseline=(current bounding box.center),
            height=120,width=120,,
            restrict y to domain=0:0.1,
            ymax = 0.0008,
            bar width=30000,
            xtick={100000,200000,300000,400000,500000},
            ytick = {0,0.0001, 0.0003,  0.0005},
            yticklabels = {0,$10^{-4}$, $3 \cdot 10^{-4}$, $5 \cdot 10^{-4}$},
            scaled y ticks = false,
            ]
            
            \addplot+[error bars/.cd,
            y dir=both,y explicit]
            table [y error plus=y-max, y error minus=y-min] {\GkexpSmallUnif};
            \addplot+[error bars/.cd,
            y dir=both,y explicit]
            table [y error plus=y-max, y error minus=y-min] {\GkexpSmallNorm};
            \legend{Uniform,Normal}
            \end{axis}
        \end{tikzpicture}
    \caption{Relative error vs. datastream size for $\alpha = 10^{-4}$, $\epsilon = 1$, $q = 0.5$}
    \label{fig:SmallError}
\end{subfigure}
\qquad
\qquad
\begin{subfigure}{.2\textwidth}
    \centering
    \begin{tikzpicture}
            \begin{axis}[
            baseline=(current bounding box.center),
            height=120,width=120,,
            ymax = 0.8,
            bar width=30000,
            ytick = {0,0.1,0.3,0.5},
            yticklabels = {0,0.1,0.3,0.5}, 
            xtick={100000,200000,300000,400000,500000},
            ]
            
            \addplot+[error bars/.cd,
            y dir=both,y explicit]
            table [y error plus=y-max, y error minus=y-min] {\GkexpLargeUnif};
            \addplot+[error bars/.cd,
            y dir=both,y explicit]
            table [y error plus=y-max, y error minus=y-min] {\GkexpLargeNorm};
            \legend{Uniform,Normal}
            \end{axis}
        \end{tikzpicture}
    \caption{$\dpexpgkGumb$ relative error vs. datastream size for $\alpha = 10^{-1}$, $\epsilon = 1$, $q = 0.5$}
    \label{fig:LargeError}
\end{subfigure}
\end{figure}

\begin{figure}
    \begin{subfigure}{.2\textwidth}
        \begin{tikzpicture}
            \begin{axis}[
            baseline=(current bounding box.center),
            height=120,width=120,,
            ymode = log,
            ymax = 100000,
            ytick = {0.000001,0.001,1},
            bar width=30000,
            xtick={100000,200000,300000,400000,500000},
            ]
            
            \addplot+[error bars/.cd,
            y dir=both,y explicit]
            table [y error plus=y-max, y error minus=y-min] {\GkexpLargeUnif};
            \addplot+[error bars/.cd,
            y dir=both,y explicit]
            table [y error plus=y-max, y error minus=y-min] {\FullLargeUnif};
            \legend{$\dpexpgkGumb$,$\dpexpfull$}
            \end{axis}
        \end{tikzpicture}
        \caption{Mean relative error and 80\% confidence interval vs. stream length over $100$ trials}
        \label{fig:LargeUnifError}
    \end{subfigure}
    \hfill
    \begin{subfigure}{.2\textwidth}
        \begin{tikzpicture}
            \begin{axis}[
            baseline=(current bounding box.center),
            height=120,width=120,,
            ymode = log,
            ybar=0pt,
            ymax = 1000000000,
            ytick = {10,1000,100000,10000000},
            bar width=30000,
            xtick={100000,200000,300000,400000,500000},
            xtick pos=left,
            ytick pos=left,
            ]
            
            \addplot+[error bars/.cd,
            y dir=both,y explicit]
            table {\GkexpLargeUnifSizes};
            \addplot+[error bars/.cd,
            y dir=both,y explicit]
            table {\FullLargeUnifSizes};
            \legend{$\dpexpgkGumb$,$\dpexpfull$}
            \end{axis}
        \end{tikzpicture}
        \caption{Data structure sizes vs. stream length \newline}
        \label{fig:LargeUnifSize}
    \end{subfigure}
    \caption{$\dpexpgkGumb$ versus $\dpexpfull$ for uniformly random data, $\alpha = 10^{-1}$, $\epsilon = 1$, $q = 0.5$}
    \label{fig:LargeUnif}
\end{figure}
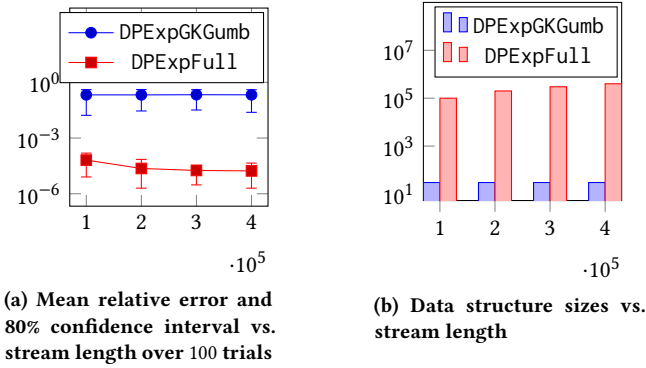
\begin{figure}
    \begin{subfigure}{.2\textwidth}
    \centering
    \begin{tikzpicture}
        \begin{axis}[
        height=120,width=120,,
        ymode = log,
        ymax = 10000,
        ytick = {0.000001,0.001,1},
        bar width=30000,
        xtick={100000,200000,300000,400000,500000},
        ]
        \addplot+[error bars/.cd,
            y dir=both,y explicit]
            table [y error plus=y-max, y error minus=y-min] {\GkexpLargeNorm};
            \addplot+[error bars/.cd,
            y dir=both,y explicit]
            table [y error plus=y-max, y error minus=y-min] {\FullLargeNorm};
    \legend{$\dpexpgkGumb$,$\dpexpfull$}
    \end{axis}
\end{tikzpicture}
    \caption{Mean relative error and 80\% confidence interval vs. stream length over $100$ trials}
    \label{fig:LargeNormError}
\end{subfigure}
    \hfill
    \begin{subfigure}{.2\textwidth}
        \begin{tikzpicture}
            \begin{axis}[
            baseline=(current bounding box.center),
            ymode = log,
            height = 120,
            width = 120,
            ybar=0pt,
            restrict y to domain=0:500000,
            bar width=30000,
            ymax = 1000000000,
            ytick = {10,1000,100000,10000000},
            xtick={100000,200000,300000,400000,500000},
            xtick pos=left,
            ytick pos=left,
            ]
            
            \addplot+[error bars/.cd,
            y dir=both,y explicit]
            table {\GkexpLargeNormSizes};
            \addplot+[error bars/.cd,
            y dir=both,y explicit]
            table {\FullLargeNormSizes};
            \legend{$\dpexpgkGumb$,$\dpexpfull$}
            \end{axis}
        \end{tikzpicture}
        \caption{Data structure sizes vs. stream length \newline}
        \label{fig:LargeNormSize}
    \end{subfigure}
    \caption{$\dpexpgkGumb$ versus $\dpexpfull$ for normally distributed data, $\alpha = 10^{-1}$, $\epsilon = 1$, $q = 0.5$ }
    \label{fig:LargeNorm}
\end{figure}




In Figures~\ref{fig:SmallErrorEps} and~\ref{fig:LargeErrorEps}, we also vary the privacy
parameter. In the small approximation factor setting we see the inverse tendency of accuracy 
with privacy which is characteristic of most DP algorithms. On the other hand, in the large approximation setting
(and indeed even in the small approximation setting for intermediate privacy parameter values), there is no such clear drop
in performance with more privacy. This motivates the question of determining the true interplay between the approximation 
factor $\alpha$ and the private parameter $\epsilon$
\arxiv{, as discussed further in Section~\ref{sec:conclusion}}.

As the main takeaway, and as predicted by the theoretical results, our algorithm performs well in practical settings where one wishes to estimate some quantity across all data items privately and using small space. For example, our results indicate that choosing an approximation factor of $\alpha = 10^{-4}$ induces an error which is also of order about $10^{-4}$ for privately computing parameters chosen according to a uniform or normal distribution, all while saving orders of magnitude in the space complexity. Another takeaway is that choosing a large approximation factor like $\alpha=0.1$ can cause large errors and is therefore not advisable in practice.


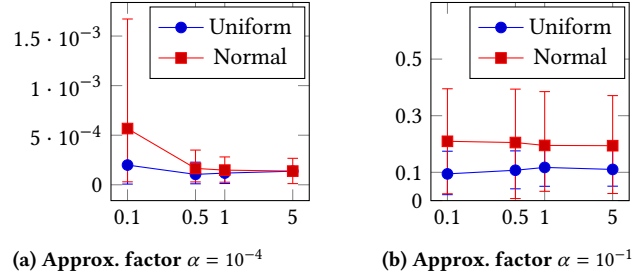
\begin{figure}
    \begin{subfigure}{.2\textwidth}
    \centering
    \begin{tikzpicture}
        \begin{axis}[
        height=120,width=120,,
        restrict y to domain=0:0.1,
        scaled y ticks = false,
        xticklabels={0.1,0.5,1,5},
        xtick={0.1,0.5,1,5},
        xmode = log
        ]
        \addplot+[error bars/.cd,
            y dir=both,y explicit]
            table [y error plus=y-max, y error minus=y-min] {\GkexpSmallUnifEps};
            \addplot+[error bars/.cd,
            y dir=both,y explicit]
            table [y error plus=y-max, y error minus=y-min] {\GkexpSmallNormEps};
    \legend{Uniform,Normal}
    \end{axis}
\end{tikzpicture}
    \caption{Approx. factor $\alpha = 10^{-4}$}
    \label{fig:SmallErrorEps}
\end{subfigure}
    \hfill
    \begin{subfigure}{.2\textwidth}
    \centering
    \begin{tikzpicture}
        \begin{axis}[
        height=120,width=120,,
        ymin = 0,
        ymax = 0.7,
        xticklabels={0.1,0.5,1,5},
        xtick={0.1,0.5,1,5},
        yticklabels={0,0.1,0.3,0.5},
        ytick={0,0.1,0.3,0.5},
        xmode = log
        ]
        \addplot+[error bars/.cd,
            y dir=both,y explicit]
            table [y error plus=y-max, y error minus=y-min] {\GkexpLargeUnifEps};
            \addplot+[error bars/.cd,
            y dir=both,y explicit]
            table [y error plus=y-max, y error minus=y-min] {\GkexpLargeNormEps};
        \legend{Uniform,Normal}
        \end{axis}
    \end{tikzpicture}
    \caption{Approx. factor $\alpha = 10^{-1}$}
    \label{fig:LargeErrorEps}
\end{subfigure}
    \caption{Relative error versus privacy parameter}
\end{figure}


\subsection{Real-World Datasets}

\begin{figure}
\begin{subfigure}{0.2\textwidth}
  \centering
  \begin{tikzpicture}
        \begin{axis}[
        height=120,width=120,,
        ymax = 1e7,
        ymin = 1e-6,
        xticklabels={0.1,0.5,1,5},
        xtick={0.1,0.5,1,5},
        ytick={1e-8,1e-4,1e-2,1},
        xmode = log,
        ymode = log,
        ]
            \addplot+[error bars/.cd,
            y dir=both,y explicit]
            table [y error plus=y-max, y error minus=y-min] {\TaxiTwo};
            \addplot+[error bars/.cd,
            y dir=both,y explicit]
            table [y error plus=y-max, y error minus=y-min] {\TaxiFive};
            \addplot+[error bars/.cd,
            y dir=both,y explicit]
            table [y error plus=y-max, y error minus=y-min] {\TaxiFull};
        \legend{1E-2,1E-5,$\dpexpfull$}
        \end{axis}
    \end{tikzpicture}
    \caption{Mean relative error and 90\% CI versus privacy parameter $\epsilon$ for $\dpexpgkGumb$
    \arxiv{with approximation parameter values $\alpha = 10^{-2}$ and $\dots,10^{-5}$, and $\dpexpfull$}.}
\label{fig:taxierror}
\end{subfigure}
\qquad
\begin{subfigure}{0.2\textwidth}
\centering
  \begin{tikzpicture}
    \begin{axis}  
    [ybar,
    height=120,width=120,,
    enlargelimits=0.15,  
    symbolic x coords={1E-2,1E-5, $\dpexpfull$}, 
    xtick=data, 
    ymode = log,
    xtick pos=left,
    ytick pos=left,
    ]  
    \addplot coordinates {($\dpexpfull$,1710672)
    (1E-2,1929)
    (1E-5,940836)};  
\end{axis}  
\end{tikzpicture}
\caption{Space usage of $\dpexpgkGumb$
\arxiv{for approximation parameter values $\alpha = 10^{-2}$ and $10^{-5}$ in $\dpexpgkGumb$, and $\dpexpfull$}.}
\label{fig:taxispace}
\end{subfigure}
\caption{Taxi data set}
\end{figure}
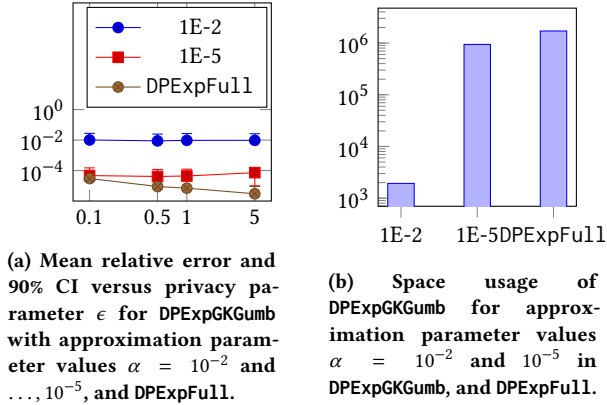

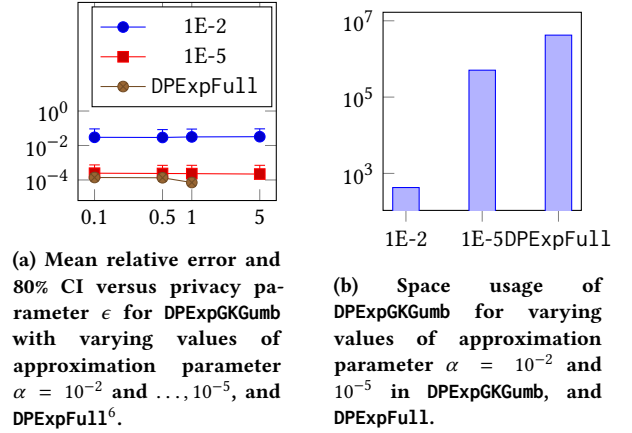
\begin{figure}
\begin{subfigure}{0.2\textwidth}
  \centering
  \begin{tikzpicture}
        \begin{axis}[
        height=120,width=120,,
        ymax = 2000000,
        xticklabels={0.1,0.5,1,5},
        xtick={0.1,0.5,1,5},
        ytick={1e-4, 1e-2, 1},
        xmode = log,
        ymode = log,
        ]
            \addplot+[error bars/.cd,
            y dir=both,y explicit]
            table [y error plus=y-max, y error minus=y-min] {\EthCOTwo};
            \addplot+[error bars/.cd,
            y dir=both,y explicit]
            table [y error plus=y-max, y error minus=y-min] {\EthCOFive};
            \addplot+[error bars/.cd,
            y dir=both,y explicit]
            table [y error plus=y-max, y error minus=y-min] {\EthCOFull};
        \legend{1E-2,1E-5,$\dpexpfull$}
        \end{axis}
    \end{tikzpicture}
    \caption{Mean relative error and 80\% CI versus privacy parameter $\epsilon$ for $\dpexpgkGumb$ 
    \arxiv{with varying values of approximation parameter $\alpha = 10^{-2}$ and $\dots,10^{-5}$, and $\dpexpfull$}\footnotemark.
    }
\label{fig:gaserror}
\end{subfigure}
\qquad
\begin{subfigure}{0.2\textwidth}
\centering
  \begin{tikzpicture}
    \begin{axis}  
    [ybar,
    height=120,width=120,,
    enlargelimits=0.15, 
    symbolic x coords={1E-2,1E-5,$\dpexpfull$}, 
    xtick=data, 
    ymode = log,
    xtick pos=left,
    ytick pos=left
    ]  
    \addplot coordinates {($\dpexpfull$,4208262)
    (1E-2,426)
    (1E-5,507585)};  
\end{axis}  
\end{tikzpicture}
\caption{Space usage of $\dpexpgkGumb$ \arxiv{for varying values of approximation parameter $\alpha = 10^{-2}$ and $10^{-5}$ in $\dpexpgkGumb$, and $\dpexpfull$}.}
\label{fig:gasspace}
\end{subfigure}
\caption{Gas sensor data set}
\end{figure}

In Table~\ref{tab:attributes}, 
we show the properties of attributes available from the taxi service and
gas sensor datasets. We pick a real-valued attribute from each dataset (the TIMESTAMP and the first ETHYLENE\_CO gas sensor value, respectively)
and calculate the median on these datasets.
In Figures~\ref{fig:taxispace} and~\ref{fig:gasspace},
for different values of the approximation parameter
$\alpha$, we show the space usage incurred on the
gas sensor and taxi cab datasets. The larger
the approximation factor, the larger the space savings with
$\dpexpgk$.
Comparing $\dpexpfull$ to $\dpexpgk$, 
we see space savings of 2 times up to 1000 times as we vary the
approximation factor. 
These results are consistent with our expectations that the
space savings are inversely proportional to the allowed approximation factor. 

\begin{table}[t]
\begin{tabular}{| l | l | l | l |}\hline
                     &  \# of Attributes       &   \# Instances  & Attribute Type  \\\hline
Taxi Service    &  9            &  1,710,671       &  Real \\\hline
Gas Sensor      &  19           &  4,178,504           & Real    \\\hline

\end{tabular}
\caption{Comparing properties of the datasets.}
\label{tab:attributes}
\end{table}

\footnotetext{The 80\% CI for error incurred by $\dpexpfull$ was entirely supported on the point $0$ and drops off axis as we use the $\log$ scale.}

\arxiv{
\subsection{Full Space Quantile Computation}
\label{sec:fullspace}
Without the bounded space requirement (i.e., space sublinear in
the stream length), we can use the exponential mechanism with
a utility function that uses the entire stream of values $X$. In that case, the
sensitivity of the utility function is at most 1. We use this
as one of the baselines for our experimental validation.

\begin{lemma}
Given any insertion only stream
$$
X = (x_1, x_2, \ldots, x_{n-1}, x_n),
$$
the sensitivity of the utility
function $u$ (under swap differential privacy)
is at most 1. $i.e., \Delta_u \leq 1$.
The function $u$ is defined as
$u(X, e) = - |\rank(X, e) -  r|$
where $r$ is the approximate $\lfloor q\cdot n\rfloor$ rank of the 
sketch and
$\rank(X, e)$ is the rank of $e$ amongst all values in the stream
$X$.
\end{lemma}

\begin{proof}

Let $|X| = n$. The utility function
becomes $-|\rank(X, e) - n_q|$ where 
$n_q = \lfloor q\cdot n\rfloor$.
Consider two streams with only one element changed: $X, X'$, denoting the element by $x_d$. Then
at time $d\leq n$, in the second stream $x_d'$ is 
inserted instead of $x_d$. In both cases, $n_q$ changes by
at most $q$ (in the case of add-remove DP) and for swap
DP, $n_q$ remains the same.
And for any $e$,
$\rank(X', e)$ would differ from
$\rank(X, e)$ by at most 1 since the rank of any element can
change by at most 1 after adding, deleting, or replacing an
item in the stream. Furthermore, for any $n\geq d$, the rank of
any $e$ will differ in $X, X'$ by at most 1 replacing $x_d$
with $x_d'$ can displace the rank of any element by at most 1.
Also, the term $n_q$ will remain the same.
(Note that in the add-remove privacy definition $n_q = \lfloor q\cdot n\rfloor$
would change to either
$\lfloor q\cdot (n+1)\rfloor$ or $\lfloor q\cdot (n-1)\rfloor$.)

The ``reverse triangle inequality'' says that for any real numbers
$x$ and $y$, $|x - y| \geq ||x| - |y||$.
As a result,
$-|\rank(X, e) - n_q| + |\rank(X', e) - n_q| \leq |\rank(X', e) -\rank(X, e)| \leq 1$
for any $e$. 
\end{proof}
}

\section{Conclusion \& Future Work}
\label{sec:conclusion}
In this work, we presented sublinear-space and differentially private algorithms
for approximately estimating quantiles in a dataset. Our solutions are 
two-part: one based on the exponential mechanism and efficiently
implemented via the use of the Gumbel distribution; the other based on
constructing histograms. 
\arxiv{
Our algorithms are supplemented with theoretical
utility guarantees. 
Furthermore, we experimentally validate our methods on
both synthetic and real-world datasets.
Our work leaves room for further exploration in various directions. Some of these questions have been addressed (at least partially) by recent subsequent work \cite{kaplan2021note}, but for completeness, we keep them here in their original form.
\begin{description}
\item[Interplay between $\alpha$ and $\eps$:] The space complexity bounds we obtain are (up to lower order terms) inversely linear in $\alpha$ and in $\eps$. While it is either known or easy to show that such linear dependence in each of these parameters in itself is necessary, it is not clear whether the $\alpha^{-1} \eps^{-1}$ term in Theorem \ref{thm:utility_universal_bound} can be replaced with, say, $\alpha^{-1} + \eps^{-1}$. Such an improvement, if possible, seems to require substantially modifying the baseline Greenwald-Khanna sketch or adding randomness. 
\item[Alternative streaming baselines:] We base our mechanisms upon  the GK-sketch, which is known to be space-optimal among deterministic streaming algorithms for quantile approximation. The use of a deterministic baseline simplifies the analysis and the overall solution, but better randomized streaming algorithms for the same problem are known to exist. What would be the benefit of working, e.g., with the (optimal among randomized algorithms) KLL-sketch \cite{KarninLL16}? 
\item[Dependence in universe size:] The dependence of our space complexity bounds in the size of the universe, $\mathcal{X}$, is logarithmic. Recent work of Kaplan et al.~\citep{DBLP:conf/colt/KaplanLMNS20} (see also~\cite{BunNSV15}) on the sample (not space) complexity of privately learning thresholds in one dimension,
a fundamental problem at the intersection of learning theory and privacy, demonstrate a bound polynomial in $\log^* |\mathcal{X}|$ on the sample complexity. As quantile estimation and threshold learning are closely related problems, this raises the question of whether techniques developed in the aforementioned papers can improve the dependence on $|\mathcal{X}|$ in our bounds.
\item[Random order:] The results presented here (except for those about normally distributed data) all assume that the data stream is presented in worst case order, an assumption that may be too strong for some scenarios. Can improved bounds be proved when the data elements are chosen in advance but their order is chosen randomly? This can serve as a middle ground between the most general case (which we address in this paper) and the case where data is assumed to be generated according to a certain distribution.
\end{description}
}

\bibliographystyle{ACM-Reference-Format}
\bibliography{main}

\appendix

\section{Greenwald-Khanna Sketch}

For completeness of our algorithm's description, we specify the operations in the Greenwald-Khanna (GK)
non-private sketch.
Throughout, we will use $n = n(t)$ to denote the number of elements encountered up to
time $t\in\mathbb{Z}_+$.
Some of the operations outlined here will be used a subroutines for the DP
procedures. 

\subsection{The Sketch}

Let $X = (x_1, x_2, \ldots, x_n)$ be a stream of items and $S(X)$ be the resulting
sketch with size sublinear in $n$. The GK sketch stores
$$
S(X) = \langle t_0, t_1, \ldots, t_{s-1}\rangle,\quad \forall i\in\{0, \ldots, s-1\}, t_i = (v_i, g_i, \Delta_i),
$$
where 
$g_i = r_{min}(v_i) - r_{min}(v_{i-1})$ and
$\Delta_i = r_{max}(v_i) - r_{min}(v_i)$.
We reserve $v_0, v_{s-1}$ be denote the smallest and largest elements seem in the stream
$X$, respectively.
We use $S(X)[i]$ to refer to the $i$th tuple in the sketch $S(X)$.
i.e., for any $i$, $S(X)[i] = t_i = (v_i, g_i, \Delta_i)$.

Implicitly, the goal is to (implicitly) maintain bounds
$r_{min}(v)$ and $r_{max}(v)$ for every $v$ in $S(X)$.
$r_{min}(v)$ and $r_{max}(v)$ are the lower and upper bounds on the rank of $v$ amongst all
items in $X$, respectively. We can compute these bounds as follows:
$$
r_{min}(v_i) = \sum_{j\leq i} g_j,\quad
r_{min}(v_i) = \sum_{j\leq i} g_j + \Delta_i.
$$

As a result, $g_i + \Delta_i - 1$ is an upper bound on the number of items between
$v_{i-1}$ and $v_i$. In addition, $n = \sum_i g_i$.

The sketch is built in such a way to guarantee (maximum) error of 
$\max_{i=0}^{s-1} (g_i + \Delta_i)/2$ for approximately computing any quantile using the sketch.

We will also impose a tree structure over tuples in $S(X)$ (mostly because of the
merge procedure) as follows:
the tree $T(X)$ associated with $S(X)$ has a node $V_i$ for each $t_i$. The
parent of a node $V_i$ is the node $V_j$ such that $j$ is the smallest index greater than $i$
with $\band(t_j) > \band(t_i)$.

$\band(t_i)$ is the band of $\Delta_i$ at time $n$ and $\band_\tau(n)$ as all tuples 
that had band value of $\tau$. All possible values of $\Delta$ are denoted as bands and it
can take on values between\\ $(0, \frac{1}{2}2\alpha n, \frac{3}{4}2\alpha n, \ldots, \frac{2^i-1}{2^i}2\alpha n, \ldots, 2\alpha n - 1, 2\alpha n)$ corresponding to capacities of
$(2\alpha n, \alpha n, \ldots, 8, 4, 2, 1)$.

\subsection{Quantile}

Algorithm~\ref{alg:quantile} computes the $\alpha$-approximate $q$-quantile based on the
sketch $S(X)$ that has size that is sublinear in $n$.

The algorithm goes through all tuples and checks if the condition
$\max(r - r_{min}(v_i), r_{max}(v_i) - r) \leq \alpha n$ is satisfied and return
$(i, v_i)$ as the representative approximate quantile.
This algorithm will be used a subroutine for one or more of our differentially private
algorithms.

\begin{lemma}[Proposition 1 \& Corollary 1~\cite{GreenwaldK01}]
    If after receiving $n$ items in the stream, the sketch $S(X)$ satisfies the property
    $\max_i(g_i + \Delta_i)\leq 2\alpha n$, then Algorithm~\ref{alg:quantile} returns an
    $\alpha$-approximate $q$-quantile.
\end{lemma}

\begin{proof}
    The algorithm computes $r = \lceil q n\rceil$. 
    Then the condition $\max(r - r_{min}(v_i), r_{max}(v_i) - r) \leq \alpha n$ clearly is (by definition)
    an $\alpha$-approximate $q$-quantile. We still need to show that such $v_i$ always exists.
    First set $e = \max_i(g_i + \Delta_i)/2$.
    If $r > n-e$, then $r_{min}(v_{s-1}) = r_{max}(v_{s-1}) = n$ so that $i=s-1$ satisfies the
    property. When $r \leq n - e$, then the algorithm chooses the smallest index $j$ such that
    $r_{max}(v_j) > r + e$ so that $r - e \leq r_{min}(v_{j-1})$. This follows since if
    $r - e > r_{min}(v_{j-1})$ then
    $r_{max}(v_j) = r_{min}(v_{j-1}) + g_j + \Delta_j > r_{min}(v_{j-1}) + 2e$ which contradicts
    the definition of $e$.
\end{proof}

\begin{algorithm}
\KwIn{$S(X), q, n, \alpha\text{ (approximation parameter)}$}

Compute $r = \lceil q n\rceil$

\For{$i=0,\ldots, s-1$} {
    $(v_i, g_i, \Delta_i) = S(X)[i]$
    
    \If {$\max(r - r_{min}(v_i), r_{max}(v_i) - r) \leq \alpha n$} {
        \Return $(i, v_i)$
    }
}

\Return $\perp$

\caption{$\Quantile(S(X), q, n, \alpha)$: Computing $\alpha$-Approximate Quantiles}
\label{alg:quantile}
\end{algorithm}

\subsection{Insert}

Algorithm~\ref{alg:streaminsert} goes through a stream of items and inserts into the 
sketch. The algorithm calls a compress operator on the data structure
every time that $i\equiv  0 \mod\frac{1}{2\alpha}$ for any $i\in[m]$.

\begin{algorithm}
\KwData{$x_1, x_2, \ldots, x_m, \ldots$}
\KwIn{$S(X), \alpha\text{ (approximation parameter)}$}

$n = 0$

\For{$i=1, \ldots, m, \ldots$} {
   
\If {$i \equiv 0 \mod\frac{1}{2\alpha}$} {
$\Compress(S(X), \alpha, n)$
}

$\Insert(S(X), \alpha, x_i)$

$n = n + 1$
}

\Return $S(X), n$

\caption{Inserting a stream of items into Summary Sketch}
\label{alg:streaminsert}
\end{algorithm}

Algorithm~\ref{alg:insert} inserts a particular item $x_n$ into the data structure
$S(X)$. In the special case where $x_n$ is a minimum or maximum, it inserts 
the tuple $(x_n, 1, 0)$ at the beginning or end of $S(X)$. Otherwise, it finds
an index $i$ such that $v_{i-1} \leq x_n < v_i$ and then inserts the
tuple $(x_n, 1, \lfloor 2\alpha n\rfloor)$ into $S(X)$ at position $i$.

\begin{algorithm}
\KwData{$x_n$}
\KwIn{$S(X), \alpha\text{ (approximation parameter)}$}

$(v_0, g_0, \Delta_0) = S(X)[0]$

$(v_{s-1}, g_{s-1}, \Delta_{s-1}) = S(X)[s-1]$

\If {$x_n < v_0$} {
Shift all positions in $S(X)[0\ldots s-1]$ to $S(X)[i\ldots s]$

$S(X)[0] = (x_n, 1, 0)$

}

\ElseIf {$x_n > v_{s-1}$} {

$S(X)[s] = (x_n, 1, 0)$

} 
\Else {
\For{$i=0,\ldots, s-1$} {
    $(v_i, g_i, \Delta_i) = S(X)[i]$
    
    \If {$v_{i-1} \leq x_n < v_i$} {
        Shift all positions in $S(X)[i\ldots s-1]$ to $S(X)[i+1\ldots s]$
        
        $S(X)[i] = (x_n, 1, \lfloor 2\alpha n\rfloor)$
    }
}
}
\Return $S(X), s+1$

\caption{$\Insert(S(X), \alpha, x_n)$: Inserting into Summary Sketch}
\label{alg:insert}
\end{algorithm}

\subsection{Compress}

\begin{algorithm}
\KwIn{$S(X), \alpha\text{ (approximation parameter)}, n$}

\If {$n < \frac{1}{2\alpha}$} {
	\Return
}

\For{$i=s-2,\ldots, 0$} {
	$t_i = (v_i, g_i, \Delta_i) = S(X)[i]$
	
	$t_{i+1} = (v_{i+1}, g_{i+1}, \Delta_{i+1}) = S(X)[i+1]$
	
	Compute $g_i^*$, the sum of $g$-values of tuple $t_i$ and its descendants
	
	\If {$\band(t_i)\leq \band(t_{i+1}) \And (g_i^* + g_{i+1} + \Delta_{i+1} < 2\alpha n)$} {
		
		Delete all descendants of $t_i$ and the tuple $t_i$ from sketch $S(X)$
		
		Update $t_{i+1}$ in $S(X)$ to $(v_{i+1}, g_i^* + g_{i+1}, \Delta_{i+1})$
	}
}

\caption{Compressing the Sketch}
\label{alg:compress}
\end{algorithm}

The $\Compress$ operation is an internal operation used for compressing
(contiguous) tuples in $S(X)$. The goal of this operation is to merge a node and its 
descendants into either its right sibling or parent node. After merge,
we have to maintain the property that the tuple is not full. A tuple is full
when $g_i + \Delta_i \geq \lfloor 2\alpha n\rfloor$.
By Proposition~\ref{prop:compress}, a node and its children will form a 
contiguous segment. Let
$g_i^*$ be the sum of $g$-values of tuple $t_i$ and all of its descendants.
Then merging $t_i$ and its descendants would update
$t_{i+1}$ in $S(X)$ to $(v_{i+1}, g_i^* + g_{i+1}, \Delta_{i+1})$ and delete
$t_i$ and all of its descendants.

\begin{proposition}[Proposition 4 in~\cite{GreenwaldK01}]
For any node $V$, the set of all its descendants in the tree
forms a contiguous segment in $S(X)$.
\label{prop:compress}
\end{proposition}

\end{document}